\documentclass{article}

\usepackage{arxiv}

\usepackage{amsmath, amsthm, amssymb, epsfig, sectsty, graphicx, float, url}
\usepackage{bm}
\usepackage{algorithm, algorithmicx}
\usepackage[noend]{algpseudocode}

\RequirePackage{bm}
\RequirePackage{endnotes}

\usepackage{tikz}

\def\<#1>{[ #1]_\alpha}
\def\p<#1>{\langle #1 \rangle_\alpha}
\newcommand{\N}{\mathbb{N}}
\newcommand{\cM}{\mathcal{M}}

\newcommand{\Z}{\mathbb{Z}}
\newcommand{\E}{\mathbb{E}}
\newcommand{\C}{\mathbb{C}}
\newcommand{\R}{\mathbb{R}}
\newcommand{\sfM}{{\sf M}}
\newcommand{\cof}{\operatorname{cof}}
\newcommand{\adj}{\operatorname{adj}}

\newcommand{\G}{\mathcal{G}}

\newcommand{\Max}{\operatorname{Max}}
\newcommand{\Min}{\operatorname{Min}}
\newcommand{\opt}{^{\ast}}

\newcommand{\cP}{\mathcal{P}}
\newcommand{\sens}{d}

\newcommand{\alphabw}{\alpha_{\sf Bw}}
\newcommand{\rhobw}{\rho_{\sf Bw}}
\newcommand{\alphad}{\alpha_{\sf d}}
\newcommand{\rhod}{\rho_{\sf d}}

\theoremstyle{plain}
\newtheorem{theorem}{Theorem}[section]
\newtheorem{lemma}[theorem]{Lemma}

\newtheorem{corollary}[theorem]{Corollary}
\newtheorem{proposition}[theorem]{Proposition}

\newtheorem{assumption}[theorem]{Assumption}
\theoremstyle{definition}
\newtheorem{definition}[theorem]{Definition}
\newtheorem{remark}[theorem]{Remark}

\usepackage{cleveref}
\usepackage[colorinlistoftodos,bordercolor=orange,backgroundcolor=orange!20,linecolor=orange,textsize=scriptsize]{todonotes}

\usepackage{booktabs}       
\usepackage{tabularx}
\usepackage{makecell}
\usepackage{array}

\usepackage{mathtools}

\begin{document}

\title{Sharper bounds on the thresholds for sensitive and Blackwell optimality}

\author{
Xavier Allamigeon \\
INRIA and CMAP \\
 Ecole Polytechnique \\    \texttt{xavier.allamigeon@inria.fr}
\And
St{\'e}phane Gaubert \\
INRIA and CMAP\\
 Ecole Polytechnique \\    \texttt{stephane.gaubert@inria.fr}
\And
Julien Grand-Cl{\'e}ment \\
ISOM Department\\
HEC Paris \\
 \texttt{grand-clement@hec.fr} \\
   \And
  Ricardo Katz\\
CIFASIS-CONICET\\
  \texttt{katz@cifasis-conicet.gov.ar} \\
}

\maketitle

\begin{abstract}
In perfect-information two-player stochastic games, the notions of Blackwell and sensitive optimality provide generalizations of the classical mean-payoff optimality and discount optimality criteria to account for more farsighted preferences. We provide bounds on the Blackwell threshold $\alphabw$ and the $d$-sensitive thresholds $\alphad$, defined as the smallest discount factors above which discount optimal policies coincide with Blackwell optimal policies and $d$-sensitive optimal policies respectively. Our refined bounds improve upon prior work by focusing on ``reduced'' families of polynomials and, crucially, our bounds are tight in terms of controlling the degrees and heights of the minimal polynomials of the Blackwell thresholds. We apply classical root separation based on Mahler's and Cauchy's bounds to our reduced families to derive the strongest upper and lower bounds on $\alphabw$ in the literature, and we are the first to obtain bounds on $\alphad$ in the multichain stochastic setting.
\end{abstract}
\section{Introduction.}
In this paper we focus on two-player zero-sum stochastic games with perfect information, a particular case of Shapley's stochastic games~\cite{shapley1953stochastic} where each state is controlled by exactly one of the two players. 
The most widely studied optimality criterion is discount optimality, where a pair of policies is discount optimal if the players optimize the return obtained over an infinite horizon and future rewards are discounted with a discount factor $\alpha \in (0,1)$. 
Many refinements of discount optimality exist in order to compute policies that can be more farsighted, including {\em $d$-sensitive optimality} (introduced for MDPs with $d \in \{-1,0,...\}$ by~\cite{Veinott1969}), which generalizes the classical mean-payoff optimality (for $d=-1$) and bias optimality (for $d=0$), as well as {\em Blackwell optimality}, i.e., policies that remain discount optimal for all values of the discount factor sufficiently close to $1$~\cite{blackwell1962discrete}. It remains an open question to design polynomial-time algorithms for computing Blackwell optimal policies and $d$-sensitive optimal policies. This difficulty arises because these optimality criteria model preferences where future rewards count as much as current rewards, so that the discount factor $\alpha$ is equal to $1$  and the Shapley operators may not be contractions. In contrast, many efficient algorithms exist to compute discount optimal policies. Strategy iteration computes discount optimal strategies in a number of iterations in $\tilde{O}((1-\alpha)^{-1})$~\cite{hansen2013strategy}. When the game parameters are known, and if the game parameters are unknown some model-free methods scale as $\tilde{O}((1-\alpha)^{-3})$, e.g., ~\cite{sidford2020solving}. In summary, we note that the complexity of computing discount optimal policies is controlled by the term $(1-\alpha)^{-1}$.

The relation between these three optimality criteria has been well-studied in the literature.~\cite{Veinott1969} shows (for MDPs) that Blackwell optimal policies correspond to policies that are $d$-sensitive optimal for all $d$.
Additionally, the analysis from~\cite{liggettlippman} shows that for a certain {\em Blackwell threshold} $\alphabw \in [0,1)$, any discount optimal pair of policies for $\alpha>\alphabw$ is also Blackwell optimal. The authors in~\cite{gaubert2026thresholds} define analogously the {\em $d$-sensitive threshold} $\alphad \in [0,1)$ as the smallest threshold above which any discount optimal policy is also $d$-sensitive optimal. The computational advantage of knowing $\alphabw$ and $\alphad$ (or any upper bound on these quantities) is straightforward: to compute Blackwell or $d$-sensitive optimal policies, we then just need to compute discount optimal policies for $\alpha>\alphabw$ or $\alpha > \alphad$, and this can be done in a complexity controlled by $(1-\alpha)^{-1}$ as explained in the previous paragraph. This is the starting point of this paper, where our goal is to give upper bounds on $\alphabw$ and $\alphad$, based on the game parameters: the number of states $n \in \N$, the maximum absolute value of the integer rewards $W \in \N$, and the common denominator of the transition probabilities $M \in \N$ (see Assumption~\ref{assumption} and discussion thereafter). 

There have been several attempts to bound $\alphad$ and $\alphabw$ in the literature.  The existence of $\alphad$ is shown in the recent work~\cite{gaubert2026thresholds} for games and in~\cite{Veinott1969} for MDPs. 
The existence of $\alphabw$ is from~\cite{liggettlippman} (for games) and ~\cite{blackwell1962discrete} (for MDPs). 
In terms of upper bounds, the works~\cite{andersson2009complexity,grand2026beyond,asadi2024deterministic,gaubert2026thresholds,mukherjee2025howard} provide upper bounds on $\alphabw$ by constructing polynomials that cancel at $\alphabw$ and at $1$, and use results from number theory to separate two roots of the same polynomial. 
This results in upper bounds on $(1-\alphabw)^{-1}$ that are at worst in the order of $\exp(c \sqrt{n} \log(n))$ for some $c>0$ (possibly dependent on $M$ and $W$). In terms of lower bounds, \cite{boros2013discounted,mukherjee2025efficient} show that $(1-\alphabw)^{-1}$ may be in the order of $2^{n/3}$ with $n \in \N$ the number of states but these results require non-deterministic instances, while \cite{asadi2024deterministic} shows that even for deterministic instances $(1-\alphabw)^{-1}$ may be in the order of $\exp(c \; n^{1/5}\sqrt{\log(n)})$ for some constant $c>0$.
Regarding $d$-sensitive thresholds, \cite{zwick} shows that $(1-\alpha_{-1})^{-1}$ is bounded by $8Wn^3$ for deterministic games. 
\cite{akian2013policy} and \cite{boone2023discounted} show bounds on $(1-\alpha_{-1})^{-1}$ for SGs and MDPs but these bounds are not expressed in terms of the game parameters (e.g. evaluating the bound in~\cite{boone2023discounted} requires computing mean-payoff optimal policies). Finally, \cite{gaubert2026thresholds} shows bounds on $\alphad$ but they require a unichain assumption for stochastic instances. We provide a more detailed literature review in Section~\ref{sec:discussion}.

Overall, many upper and lower bounds have been obtained on $\alphabw$ and $\alphad$ in the literature, without a clear understanding of their tightness, both in terms of the size (degree and height) of the polynomials used in root separations and in terms of the resulting bounds on the thresholds. We present our bounds for the corresponding {\em interest rates} $\rhobw$ and $\rhod$, defined such that $\alphabw = (1+\rhobw)^{-1}$ and $\alphad = (1+\rhod)^{-1}$; so upper bounding $\alphabw$ and $\alphad$ is equivalent to lower bounding $\rhobw$ and $\rhod$ (and vice versa for lower bounds on discount factors).
Our {\bf main contributions} are as follows.
\begin{enumerate}
    \item {\bf Tight control of reduced families.} We formalize the notion of {\em comparison family}, which is a family of polynomials that control the signs of the difference of two value functions. While all prior works use the comparison family directly induced by the difference of value functions across all pairs of policies, we introduce a ``reduced" family of polynomials to bound the Blackwell and $d$-sensitive thresholds, based on Veinott's comparison vectors for two adjacent policies that differ at only one state. A careful application of the Matrix Forest Theorem enables us to provide bounds on the degree and height of our polynomials that significantly improve upon prior approaches, and we show that our bounds on the degree and height of the minimal polynomial of $\rhobw$ are tight. 
    \item {\bf New upper and lower bounds on the Blackwell threshold.} We then apply root separation results to our reduced family to bound $\rhobw$. Because our comparison family is much better behaved than the ones considered in the past literature, we obtain lower bounds on $\rhobw$ (yielding upper bounds on $\alphabw$) that compare favorably with the ones obtained in prior work. We complement these results by obtaining new, improved upper bounds for $\rhobw$ (yielding lower bounds on $\alphabw$): our hard-instance bound for stochastic instances is tight in terms of order of $M$, while our hard-instance bound for deterministic instances improves upon prior work but leaves a gap open.
    \item {\bf Bounds for $d$-sensitive thresholds.} We provide the first bound for the $d$-sensitive thresholds in the multichain stochastic setting (\Cref{Th:d-sensitiveBound}), based on a careful analysis of the recurrent classes of adjacent policies. We also show $\bar{d}$-sensitive optimality suffices for Blackwell optimality with $\bar{d}=O\left(\sqrt{n(1+\log(W))}\right)$, a significant improvement over the previously known bounds. Our bounds on $\rhod$ for deterministic instances also greatly improve upon prior work. 
\end{enumerate}
In summary, our work clarifies the existing proof techniques for bounding the Blackwell and $d$-sensitive thresholds, shows that we can construct a comparison family with substantially smaller degree and height than previously built, and provides tight bounds in terms of degree and height of the minimal polynomial of $\rhobw$. Our bounds on the Blackwell and $d$-sensitive thresholds obtained from our reduced family compare favorably with all prior work.

{\bf Outline.} We briefly review perfect-information stochastic games in \Cref{sec:preliminaries}. In \Cref{sec:comparison poly}, we introduce the concept of a reduced family of comparison polynomials, show how these families can be used to bound the Blackwell threshold $\alphabw$, and derive bounds on their coefficients and degree. Our bounds on the Blackwell thresholds are presented in \Cref{sec:separation blackwell}, and  in \Cref{sec:bound d sensitive} we focus on the $d$-sensitive thresholds $\alphad$. The implications of our results are discussed in \Cref{sec:discussion}. All proofs omitted from the main body are provided in the appendix.

\section{Preliminaries.}\label{sec:preliminaries}
\subsection{Perfect-information stochastic games.}
We focus on two-player zero-sum stochastic games. We assume that the set of states and the set of actions of each player are finite. We write $[n]=\{1,\dots,n\}$ for the set of states,  where $n \in \N$. At every period $k \in \N$, the game is in a state $i_k$ observable by both players, called $\Min$ and $\Max$. If player $\Min$ chooses action $a_k$ and player $\Max$ chooses action $b_k$, an instantaneous reward $r_{i_k}^{a_k b_k} \in \R$ is incurred by player $\Min$ and obtained by player $\Max$, and then the game transitions to the next state $i_{k+1}$ with probability $P^{a_kb_k}_{i_ki_{k+1}´} \in [0,1]$.

In {\em perfect-information} stochastic games, the state space $[n]$ can be divided into a set of states controlled by player $\Min$ and a set of states controlled by player $\Max$. 
A {\em policy} of a player is a map from the set of all finite-horizon histories to a (possibly randomized) choice of an action. Let $(\sigma,\tau)$ be a pair of policies of the players and $i \in [n]$ be the initial state, i.e., the state where the game starts at period $k=1$. Recall that $(\sigma,\tau)$ induces a probability measure on the set of infinite-horizon histories. The {\em discounted value function} $\alpha \mapsto v^{\sigma,\tau}_{i}(\alpha)$ is the function from $(0,1)$ to $\mathbb{R}$ representing the payoff paid by player $\Min$ to player $\Max$:
\begin{equation}\label{eq:discounted value functions}
v^{\sigma,\tau}_{i}(\alpha) := \E^{\sigma,\tau} \left[ \sum_{k=1}^{+\infty} \alpha^k r_{i_k}^{a_k b_k} \; | \; i_1 = i \right]
\end{equation}
where $\left(r_{i_k}^{a_k b_k}\right)_{k \in \N}$ is the random sequence of instantaneous rewards induced by $(\sigma,\tau)$. 
\begin{definition}[Discount and Blackwell optimality]\label{def:discount optimality}
A pair of policies $(\sigma^*,\tau^*)$ of players $\Min$ and $\Max$ is {\em discount optimal} for the discount factor $\alpha$ if
\begin{equation}\label{eq:DefDiscOpt}
v^{\sigma,\tau\opt}_{i}(\alpha) \geq v^{\sigma\opt,\tau\opt}_{i}(\alpha) \geq v^{\sigma\opt,\tau}_{i}(\alpha) 
\end{equation}
for any state $i$ and any policies $\sigma$ and $\tau$ of players $\Min$ and $\Max$ respectively.

A pair of policies $(\sigma^*,\tau^*)$ is {\em Blackwell optimal} if there exists $\bar{\alpha} <1$ such that $(\sigma^*,\tau^*)$ is discount optimal for all discount factors larger than $\bar{\alpha}$.
\end{definition}
The existence of stationary discount optimal policies was established in the seminal work of Shapley~\cite{shapley1953stochastic}, where it is also shown that these policies can be taken to be {\em deterministic} under the additional assumption that the game is perfect-information.
The existence of stationary deterministic Blackwell optimal policies for perfect-information SGs is a consequence of the analysis in the seminal paper~\cite{liggettlippman}. Thus, in what follows, unless stated otherwise and without loss of generality, the policies we consider are stationary deterministic. We now  define the Blackwell threshold. 
\begin{definition}[Blackwell threshold]\label{def:blackwell threshold}
The Blackwell threshold  $\alphabw$ is the smallest scalar $\alpha' \in [0,1)$ such that for any discount factor $\alpha \in (\alpha',1)$, any pair of discount optimal policies for the discount factor $\alpha$ is also Blackwell optimal.
\end{definition}
The existence of the Blackwell threshold is shown in the proof of Theorem 1 of~\cite{liggettlippman}, by showing that discounted value functions are rational functions in the discount factor, so discounted value functions associated with different pairs of stationary policies can only intersect finitely many times in $[0,1)$, and hence the difference of discounted value functions cannot change sign in a neighborhood of $1$. The same idea is present in~\cite{blackwell1962discrete} for MDPs. This highlights a crucial connection between the Blackwell threshold and roots of polynomials, which we will leverage in the next section.
We are also interested in the notion of {\em d-sensitive} optimality~\cite{gaubert2026thresholds}, which provides an intermediary notion of ``farsightedness" between discount optimality and Blackwell optimality.
\begin{definition}[$\sens$-sensitive Optimality]\label{def:d sensitive optimality}
Given $d \in \{-1,0,1,\ldots\}$, a pair of policies $(\sigma\opt,\tau\opt)$ of players Min and Max is {\em $\sens$-sensitive optimal} if
\[
\liminf_{\alpha\to 1^-} \; (1-\alpha)^{-\sens}\left(v^{\sigma,\tau\opt}_{i}(\alpha)-v^{\sigma\opt,\tau\opt}_{i}(\alpha) \right) \geq 0 ,\;
 \liminf_{\alpha\to 1^-} \; (1-\alpha)^{-\sens}\left(v^{\sigma\opt,\tau\opt}_{i}(\alpha)-v^{\sigma\opt,\tau}_{i}(\alpha) \right) \geq 0
\]
for any state $i$ and any policies $\sigma$ and $\tau$ of players Min and Max respectively.
\end{definition}
The previous definition recovers $\sens$-sensitive optimality for the one-player case (i.e., for MDPs, see~\cite{Veinott1969} and Chapter 10 in~\cite{puterman2014markov}). It is worth emphasizing that $d$-sensitive optimal policies are also $d'$-sensitive optimal for $d' \leq d$, that Blackwell optimal policies coincide with policies that are $\sens$-sensitive optimal for all $\sens \geq -1$. Mean-payoff optimal policies coincide with $d$-sensitive optimal policies for $\sens=-1$. We define the {\em $\sens$-sensitive threshold} as follows.
\begin{definition}[$\sens$-sensitive threshold]
Given $d \in \{-1,0,1, \ldots \}$, the $\sens$-sensitive threshold  $\alphad$ is the smallest scalar $\alpha' \in [0,1)$ such that for any discount factor $\alpha \in (\alpha',1)$, any pair of discount optimal policies for the discount factor $\alpha$ is also $\sens$-sensitive optimal.
\end{definition}
The existence of the $\sens$-sensitive thresholds follows from the existence of the Blackwell threshold. In particular, by definition we have $\alphad \leq \alphabw$ for any $\sens$. As explained in the introduction, if we know the value of $\alphabw$ and $\alphad$, we can compute discount optimal policies for a discount factor $\alpha$ close to 1 (larger than $\alphabw$ or $\alphad$) and recover Blackwell optimal or $d$-sensitive optimal policies. Since the complexity of the algorithms for computing discount optimal policies grows with $(1-\alpha)^{-1}$, we want to find upper bounds on $\alphabw$ and $\alphad$ that are as small as possible, i.e., as far away from $1$ as possible. 
The main goal of this paper is to obtain upper bounds on $\alphabw$ and $\alphad$, based on the following assumption.
\begin{assumption}\label{assumption}
$\Gamma$ is a perfect-information stochastic game with finitely many states and actions, integer rewards, and rational transition probabilities. The parameters of the game $\Gamma$ are $n$, $W$, and $M$, and are defined as follows:
\begin{itemize}
    \item $n \in \N$ is the number of states,
    \item $W \in \N$ is an upper bound on the absolute values of the integer rewards,
    \item $M \in \N$ is the common denominator of the transition probabilities.
\end{itemize}
\end{assumption}
Note that Assumption~\ref{assumption} is necessary to obtain bounds on the Blackwell threshold, as shown in~\cite[Proposition~4.3]{grand2026beyond}. In the case where the game can be modeled directly with rational rewards, we just need to multiply all rewards by a constant (which does not change the sets of optimal policies) to recover integral rewards. Additionally, the denominator $M$ in the transition probabilities appears naturally if one estimates the transitions as empirical frequencies from some historical dataset. Finally, we note that assuming the instance can be represented with a finite number of bits (as is classical to study polynomial-time algorithms) means that the data is rational, as in Assumption~\ref{assumption}.
\paragraph{Interest rate.}
In some cases, it is more convenient to work with an {\em interest rate} $\rho$ instead of the discount factor $\alpha$. In that case, the discount factor is given by  $\alpha \coloneqq 1 /(1+\rho)$ for $\rho>0$, so $\rho=(1-\alpha)/\alpha$, see~\cite{veinott}, with the convention that $\rho=+\infty$ if $\alpha=0$. Thus, we can define the {\em Blackwell interest rate threshold} $\rhobw$ similarly to $\alphabw$, which satisfies $\rhobw=(1-\alphabw)/\alphabw$. Note that the interest rate $\rho$ approaches $0$ when the discount factor $\alpha$ approaches $1$, and vice versa. Then, upper bounding $\alphabw$ can be done by lower bounding $\rhobw$ (since $\rho \mapsto 1/(1+\rho)$ is decreasing).
For this reason, we give our main results in terms of lower bounds on $\rhobw$, and the term {\em ``Blackwell threshold"} is used for both $\rhobw$ and $\alphabw$ (the precise meaning will be clear from context). A similar approach will be taken for the $d$-sensitive threshold $\alphad$, considering the {\em $\sens$-sensitive interest rate threshold} $\rhod$, which satisfies $\rhod = (1-\alphad)/\alphad$. 
\subsection{From stochastic games to Markov decision processes.}\label{sec:games to mdps}
Before presenting our main results, we present a reduction from SGs to MDPs, which will greatly simplify our proofs.
It is well-known that a two-player zero-sum stochastic game $\Gamma$ induces a collection of MDPs, obtained by fixing the policy of one of the players. For instance, if we fix a stationary deterministic policy $\tau$ of player $\Max$, finding an optimal best-response of player $\Min$ can be reduced to a MDP $\Gamma_{\tau}^{\sf Min}$ played by player $\Min$. The set of states of $\Gamma_\tau^{\sf Min}$ is the set of states of $\Gamma$, the set of actions in $\Gamma_\tau^{\sf Min}$ is the set of actions of player $\Min$ in $\Gamma$ in the states controlled by this player and $\{\tau(i)\}$ in each state $i$ controlled by player $\Max$, and the rewards and transitions in $\Gamma_\tau^{\sf Min}$ are those of $\Gamma$ for the actions available in $\Gamma_\tau^{\sf Min}$. If we fix a stationary deterministic policy $\sigma$ of player $\Min$, we obtain an MDP $\Gamma_{\sigma}^{\sf Max}$ played by player $\Max$ by the same construction above, except that one needs to reverse the sign of the rewards (in this paper we use the convention that MDPs are minimizing their objective function). 

We denote by $\cM(\Gamma)$ the set of MDPs induced by the game $\Gamma$, i.e., the set of all $\Gamma_{\tau}^{\sf Min}$ and all $\Gamma_{\sigma}^{\sf Max}$ for any stationary deterministic policies $\sigma$ and $\tau$ of players $\Min$ and $\Max$ respectively. With a slight abuse of notation, we will say that an MDP instance satisfies Assumption~\ref{assumption} if the ``trivial" game induced by the MDP by adding a dummy player satisfies Assumption~\ref{assumption}. We can now state our reduction from SGs to MDPs.
\begin{proposition}\label{prop:reduction games to mdps} Let $\Gamma$ be a game. The following statements hold.
\begin{enumerate}
\item[(a)] If $\Gamma$ satisfies Assumption~\ref{assumption}, then all MDPs induced by $\Gamma$ satisfy Assumption~\ref{assumption}.
\item[(b)] If the Blackwell threshold of all MDPs in $\cM(\Gamma)$ is lower bounded by $\bar{\rho}$, then the Blackwell threshold of $\Gamma$ is lower bounded by $\bar{\rho}$.
\item[(c)] Let $d \in \{-1,0,1, \ldots \}$. If the $d$-sensitive threshold of all MDPs in $\cM(\Gamma)$ is lower bounded by $\bar{\rho}$, then the $d$-sensitive threshold of $\Gamma$ is lower bounded by $\bar{\rho}$.
\end{enumerate}
\end{proposition}
The proof of \Cref{prop:reduction games to mdps} relies on checking the Blackwell and $d$-sensitive optimality of a pair $(\sigma\opt,\tau\opt)$ in the game $\Gamma$ by checking the Blackwell and $d$-sensitive optimality of $\sigma\opt$ in $\Gamma^{\sf Min}_{\tau\opt}$ and of $\tau\opt$ in $\Gamma^{\sf Max}_{\sigma\opt}$, as evident from Definitions~\ref{def:blackwell threshold} and~\ref{def:d sensitive optimality}. We provide a proof in Appendix~\ref{app:proof mdp reduction}.
\section{Complete and reduced families of comparison polynomials.}\label{sec:comparison poly}
We now introduce the main ingredients that we use to bound the Blackwell and $\sens$-sensitive thresholds for MDPs (which suffices to bound these thresholds for games by \Cref{prop:reduction games to mdps}). 
We will use the convention that MDPs are minimizing the return.
\subsection{Main results.}
To simplify notation, we will use the definition.
\begin{definition} For any $\rho>0$ and policy $\sigma$, we write $v^{\sigma}_{\rho} \in \R^{n}$ for the value function with interest rate $\rho$, corresponding to a discount factor of $1/(1+\rho)$: $\displaystyle v^{\sigma}_{\rho} \coloneqq v^{\sigma}\left((1+\rho)^{-1}\right)$, with $v^{\sigma}$ defined as in~\eqref{eq:discounted value functions}.
\end{definition}
\begin{remark}
We note that in our definition of discounted value functions~\eqref{eq:discounted value functions}, the reward obtained at the first iteration is discounted (the summation starts at $k=1$). This departs from the usual convention, which usually does not discount the first reward (i.e., the summation in~\eqref{eq:discounted value functions} usually starts at $k=0$). We do so because it significantly lightens the notations later in the paper, while not changing anything for the values of the Blackwell and $d$-sensitive thresholds, since the Blackwell and $d$-sensitive optimalities are based on sign comparisons of the difference of discounted value functions, and multiplying or dividing by $\alpha \in (0,1)$ does not change the sign.
\end{remark}
\begin{definition}[Comparison polynomials]
Let $\cM$ be a MDP, $\sigma'$ and $\sigma$ be two policies and $i \in [n]$ a state. 
A polynomial $\Delta(\rho)\in \mathbb{Z}[\rho]$ is a \emph{comparison polynomial} for $\sigma'$ against $\sigma$ at state $i$ if, for all $\rho > 0$, we have
\begin{equation}\label{eq:def comp poly}
\Delta(\rho)\geq 0 \iff (v_{\rho}^{\sigma'})_i \leq (v^{\sigma}_{\rho})_i  \; .
\end{equation}
\end{definition}
As their names suggest, comparison polynomials enable us to compare the ordering of discounted value functions of different policies by inspecting their signs.
\begin{definition}[Comparison family and complete family]
\begin{enumerate}
\item A {\em comparison family} is a set of polynomials containing only comparison polynomials.
\item A comparison family is {\em complete} if it contains one comparison polynomial for every ordered pair of policies and every state.
\end{enumerate}
\end{definition}
It has been observed repeatedly that the Blackwell threshold $\rhobw$ is bounded from below by the smallest positive root of any nonzero polynomial in an arbitrary complete family of comparison polynomials, see e.g.~\cite{grand2023reducing} and \cite{andersson2009complexity}. By~\eqref{eq:def comp poly}, below this threshold value functions cannot intersect, and discount optimal policies remain optimal for all smaller interest rate. In fact, we will prove  that it is enough to work with a smaller class of families of polynomials.
\begin{definition}
\begin{enumerate}
\item Two policies $\sigma$ and $\sigma'$ are {\em adjacent at state $i$}, or simply {\em adjacent}, if they coincide everywhere except at state $i$.
\item A comparison family is {\em reduced} if it contains exactly one comparison polynomial for every ordered pair of adjacent policies. 
\end{enumerate}
\end{definition}
Our first main result is to show that reduced families suffice for bounding the Blackwell threshold.
\begin{theorem}\label{th:reduced family root bound alphabw}
Consider any reduced family $\cP$ of comparison polynomials for a given MDP and assume $\rhobw<+\infty$. Then, the Blackwell threshold $\rhobw$ is a positive root of a polynomial belonging to $\cP$, and consequently $\rhobw$ is lower bounded by the smallest positive root of the nonzero polynomials of this family.
\end{theorem}
Based on the above theorem, bounding the Blackwell threshold can be done in two steps:
\begin{itemize}
\item Step 1: build a reduced or complete family of comparison polynomials
\item Step 2: lower bound the positive roots of the polynomials in this family 
\end{itemize}
Step 2 can be done by any of the existing results on {\em root separation} for polynomials with integer coefficients. Root separation results typically yield weaker (smaller) separation when the degree and the size of the coefficients of the polynomial increase. Let us recall that the {\em height} $H(p)$ of a polynomial $p(x)=\sum^{h}_{k=0} a_k x^k$ with integer coefficients is given by
\[
H(p) \coloneqq  \max_{k \in \{ 0, \ldots , h\}} | a_k | \; .
\]
With a slight abuse of notation, we will talk of the {\em degree} and {\em height} of a comparison family for the maximum degree and height among the polynomials in this family.
To obtain sharper bounds on the Blackwell threshold, the main goal is, in Step 1, to find a family with degree and height as small as possible, and in Step 2, to use the most powerful separation result possible.
With this in mind, note that the degree $N$ and the height $H$ of a comparison family must depend on the integral data of the game, i.e., on the parameters of Assumption~\ref{assumption}. 
Our next main result in this section is the following.
\begin{theorem}\label{th:reduced family bound}
For any MDP satisfying Assumption~\ref{assumption} there exists a reduced family of comparison polynomials with degree at most $n-1$ and height at most $2 W M^{n-1} {n+1 \choose \lfloor (n+1)/2 \rfloor}$.
\end{theorem}
 Theorem~\ref{th:reduced family bound} enables one to use any root separation method to obtain lower bounds on $\rhobw$ (and therefore upper bounds on $\alphabw$). We will do so in Section~\ref{sec:separation blackwell}.
We show that our bounds on degrees and height are essentially tight in the next theorem. 
\begin{theorem}\label{th:reduced family bound - tightness}
For every $n \geq 2$ and $M,W \in \N$ with $2W-1$ prime and $M$ not divisible by $2W-1$, there exists an MDP satisfying Assumption~\ref{assumption} where the minimal polynomial of $\rhobw$ has degree  $n-1$ and height $(2W-1) M^{n-1} {n-2 \choose \lfloor (n-2)/2 \rfloor}$.
\end{theorem}
Since the reduced family of Theorem~\ref{th:reduced family bound} has degree at most $n-1$, we note that for the MDP instance of Theorem~\ref{th:reduced family bound - tightness} we must have that (a multiple of) the minimal polynomial of $\rhobw$ belongs to our reduced family. We will compare the bounds that we obtain for our reduced family with prior approaches in Section~\ref{sec:discussion}.
The rest of this section is devoted to the proof of Theorems~\ref{th:reduced family root bound alphabw}, \ref{th:reduced family bound} and~\ref{th:reduced family bound - tightness}.
\subsection{Veinott's comparison vectors.}\label{sec:veinott comparison}
In this section we introduce Veinott's comparison vectors, which provide some necessary and sufficient conditions for discount optimality that involve comparing a policy to its adjacent policies (instead of comparing a policy to all other policies, as in \Cref{def:discount optimality}). This will greatly simplify the later expositions of our proofs. 
\paragraph{Matrix notations.}
Given a real matrix $C$, we write $C\geq 0$ if all the entries of $C$ are nonnegative. We denote by $C_i$ the $i$-th row (or coordinate, if $C$ is a column vector) of $C$.
We denote by $r^\sigma\in\R^n$ and $P^\sigma\in \R^{n\times n}$
the reward vector and transition matrix associated with a stationary policy
$\sigma$, respectively. We will need the following notions.
\begin{definition}\label{def:laplacian resolvent}
We define $L^{\sigma} \in \R^{n \times n}$ for the {\em signed Laplacian matrix}: $L^{\sigma} \coloneqq P^{\sigma} - I$. Its resolvent $R^{\sigma}(\rho) \in \R^{n \times n}$ is $R^{\sigma}(\rho) \coloneqq (\rho I - L^{\sigma})^{-1}$. Besides, we set $J^{\sigma} \coloneqq M L^{\sigma}$ and $D^{\sigma}(\rho) \coloneqq \det \left( \rho M I - J^{\sigma}\right)$. 
\end{definition}
The next proposition highlights a few important properties of the resolvent.
\begin{proposition}\label{prop:resolvent properties - basic}
For any stationary policy $\sigma$ we have:
\begin{enumerate}
\item[(a)] $R^{\sigma}(\rho) \geq 0$ for any $\rho >0$.
\item[(b)] $R^{\sigma}_{ii}(\rho) \geq (1+\rho)^{-1}>0$ for any $i \in [n]$ and $\rho >0$.
\end{enumerate}
\end{proposition}
\begin{proof}
    Both statements follow from the fact that $R^{\sigma}(\rho) = \sum_{k = 0}^{\infty} (1+\rho)^{-(k+1)} (P^{\sigma})^{k}$. 
\end{proof}
Recall that, given a stationary policy $\sigma$ the value function $v^{\sigma}_{\rho} \in \R^{n}$ is given by:
\begin{equation}\label{eq:ValueDiscountedGame}
v^{\sigma}_{\rho} = R^{\sigma}(\rho) r^{\sigma}.
\end{equation}
The comparison vector $C^{\sigma,\sigma'}(\rho) \in \R^{n}$ is defined next.
\begin{definition}[Comparison vectors]\label{def:Comp-Vect}
Given a pair of policies $(\sigma,\sigma')$, the \emph{comparison vector} $C^{\sigma,\sigma'}(\rho) \in \R^{n}$ is defined as
\begin{equation}\label{eq:definition comparison vector}
C^{\sigma,\sigma'}(\rho) \coloneqq r^{\sigma} + (L^{\sigma}-\rho I) v^{\sigma'}_{\rho} \; .
\end{equation}
\end{definition}
The comparison vector coincides with the reduced cost in the linear programming formulation of a MDP, where every stationary policy defines a basic point. 
This vector appears in the following equality,  sometimes known as the {\em performance difference lemma} (e.g. Lemma 1 of~\cite{veinott}):
\begin{equation}\label{eq:performance difference}
v^{\sigma}_{\rho} - v^{\sigma'}_{\rho} = R^{\sigma}(\rho)C^{\sigma, \sigma'}(\rho) \; .
\end{equation} 
Since the entries of $R^{\sigma}(\rho)$ are non-negative, \eqref{eq:performance difference} implies that $C^{\sigma, \sigma'}(\rho) \geq 0 \Rightarrow  v^{\sigma}_{\rho} - v^{\sigma'}_{\rho} \geq 0$. The next lemma shows that this implication is in fact an equivalence when the policies are adjacent.
\begin{lemma}\label{lem:comp vector adjacent policies prop}
Let $\sigma$ and $\sigma'$ be adjacent policies and $i \in [n]$ the state where they differ. Then,
\begin{enumerate}
\item[(a)] For $j \neq i$, we have $C^{\sigma,\sigma'}_{j}(\rho)=0$ for all $\rho >0$.
\item[(b)] $C^{\sigma,\sigma'}_{i}(\rho)$ only depends on $\sigma'$ and $\sigma(i)$, not on $\sigma(j)$ for $j \neq i$.
\item[(c)] $(v^{\sigma}_{\rho} - v^{\sigma'}_{\rho})_i = R^{\sigma}_{ii}(\rho) C^{\sigma,\sigma'}_{i}(\rho)$ for any $\rho>0$.
\end{enumerate}
\end{lemma}
Lemma~\ref{lem:comp vector adjacent policies prop} follows directly from Definition~\ref{def:Comp-Vect} and we omit its proof for conciseness. The following proposition shows that one can characterize the optimality of a policy $\sigma\opt$ by checking the sign of the comparison vectors of adjacent policies.
\begin{proposition}\label{PropVeinott}
For any stationary policy $\sigma^\ast$ and any $\rho>0$ the following statements are equivalent.
\begin{itemize}
\item[(a)] $v^{\sigma^\ast}_{\rho} \leq v^{\sigma}_{\rho}$ for every stationary policy $\sigma$.
\item[(b)] $C^{\sigma,\sigma^\ast}(\rho) \geq 0$ for every stationary policy $\sigma$ adjacent to $\sigma^\ast$.
\end{itemize}
\end{proposition}
 \Cref{PropVeinott} is equivalent to a special case
of the optimality condition in the simplex algorithm and it extends Theorem 2 of~\cite{veinott} (which is written without referring to adjacent policies). For the convenience of the reader, we provide a proof of~\Cref{PropVeinott} from first principles. 
\begin{proof}[Proof of Proposition~\ref{PropVeinott}.]
The proof proceeds in two steps. 

{\em (a) implies (b).} This proof is by contrapositive. Assume that (b) does not hold, so that there exists a stationary policy $\sigma$ differing from $\sigma^\ast$ only at some state $i$ such that we do not have $C^{\sigma,\sigma^\ast}(\rho) \geq 0$. By items (a) and (b) of Lemma~\ref{lem:comp vector adjacent policies prop}, we must have $C^{\sigma,\sigma\opt}_{i}(\rho)<0$. Item (c) of Lemma~\ref{lem:comp vector adjacent policies prop} shows that $(v^{\sigma}_{\rho} - v^{\sigma^\ast}_{\rho})_i = R^{\sigma}_{ii}(\rho) C^{\sigma,\sigma\opt}_{i}(\rho) < 0$, since $R^{\sigma}_{ii}(\rho)>0$ (item (b) of Proposition~\ref{prop:resolvent properties - basic}). 
This shows that (a) does not hold when (b) does not hold, hence (a) implies (b).
 
{\em (b) implies (a).} This proof is also by contrapositive.
 Assume that (a) does not hold, so that $v^{\sigma^\ast}_{\rho} \leq v^{\sigma}_{\rho}$ is not satisfied for some stationary policy $\sigma$. Then, by~\eqref{eq:performance difference} it follows that there exists a state $i$ such that $C^{\sigma,\sigma^\ast}_i(\rho) < 0$. We define the stationary policy $\sigma'$ by $\sigma'(i)= \sigma(i)$ and $\sigma'(j)= \sigma^\ast(j)$ for $j\neq i$, which is adjacent to $\sigma\opt$. By item (b) of Lemma~\ref{lem:comp vector adjacent policies prop}, we have $C^{\sigma',\sigma^\ast}_i(\rho) = C^{\sigma,\sigma^\ast}_i(\rho) < 0$, and so $C^{\sigma', \sigma^\ast}(\rho) \ngeqslant 0$. This shows that (b) does not hold when (a) does not hold.

\end{proof}

\subsection{Proof of Theorem~\ref{th:reduced family root bound alphabw}.}

\begin{proof}[Proof of Theorem~\ref{th:reduced family root bound alphabw}.]
Assume $\rhobw <+\infty$. By definition of $\rhobw$, we can find a sequence of stationary deterministic policies $(\sigma_\ell)_{\ell \in \N}$ and interest rates $(\rho_{\ell})_{\ell \in \N}$ with $\lim_{\ell \to + \infty} \rho_{\ell} = \rhobw$ such that for all $\ell \in \N$ we have $\rhobw < \rho_{\ell}$, and $\sigma_{\ell}$ is discount optimal for the interest rate $\rho_{\ell}$ but not Blackwell optimal. Since there are only finitely many deterministic policies, we get that we can find a stationary deterministic policy $\sigma$ that is discount optimal at $\rhobw$ but not Blackwell optimal. The same argument from below shows that we can find a stationary deterministic policy $\sigma\opt$ that is Blackwell optimal and discount optimal at $\rhobw$. Note that for $\ell$ large enough we must have that $\sigma\opt$ is not discount optimal for the interest rate $\rho_{\ell}$ (otherwise the value functions of $\sigma$ and $\sigma\opt$ would coincide for an infinite number of interest rates, and so these value functions would be equal, which is impossible given that $\sigma$ is not Blackwell optimal). Hence, by \Cref{PropVeinott} it follows that there exists a policy $\sigma'$ which is adjacent to $\sigma\opt$ at some state $i \in [n]$ such that $C^{\sigma',\sigma\opt}_{i}(\rho_{\ell})<0$ for $\ell$ large enough.   This implies that the polynomial $\Delta(\rho)$ in $\cP$ associated with $(\sigma',\sigma\opt)$ and $i \in [n]$ is such that $\Delta(\rho_{\ell})<0$ for $\ell$ large enough, and so  $\Delta(\rhobw)\leq 0$ by continuity. Besides, we have $\Delta(\rhobw) \ge 0$ by the discount optimality of $\sigma\opt$ at $\rhobw$. Thus $\Delta(\rhobw)=0$, which completes the proof.

\end{proof}

\subsection{Proof of Theorem~\ref{th:reduced family bound}.}\label{sec:proof th 2}
In this section we provide a proof of Theorem~\ref{th:reduced family bound}. The proof proceeds in two steps. 
 In Section~\ref{sec: building comparison family}, we build a reduced family based on comparison vectors, and in Section~\ref{sec:bounding height degree} we bound its height and degree.
\subsubsection{Constructing a reduced family.}\label{sec: building comparison family}
Following the analysis in the previous section, we can track if a policy is optimal by checking the signs of Veinott's comparison vectors. However, the comparison vector is not, in general, a polynomial function in the interest rate $\rho$ because it depends on the value functions, and in general the latter functions are {\em rational} functions in the interest rate, not necessarily polynomial functions. To obtain polynomial comparison families we therefore need to multiply the comparison vectors by their polynomial denominators. 
Note that by construction, $J^{\sigma}$ is a matrix of size $n \times n$ with integer coefficients. Similarly, from the definition of the characteristic polynomial we know that $D^{\sigma}(\rho)$ is a polynomial in $\rho$ with integer coefficients.
Additionally, as we will prove later, $D^{\sigma}(\rho)$ is the common polynomial appearing in the denominator of $(v_{\rho}^{\sigma})_i$ for any state $i$, and multiplying a component of the comparison vector by $D^{\sigma}(\rho)/M \rho$ yields a polynomial. This justifies the construction of our reduced comparison family as follows: 
\begin{equation}\label{eq:definition reduced family}
\cP  \coloneqq \left\{\rho \mapsto \frac{D^{\sigma'}(\rho)}{M \rho} C^{\sigma,\sigma'}_{i}(\rho)  \; | \; i \in [n], \sigma \text{ and }\sigma' \text{ are adjacent and differ at state } i\right\}\; .
\end{equation}
We first show a crucial property of the polynomials in $\cP$: their signs and zeros are controlled by the signs and zeros of the comparison vectors.
\begin{lemma}\label{lem:zeros delta c} 
The following statements hold.
\begin{enumerate}
\item[(a)] For any policy $\sigma$ we have: $D^{\sigma}(0)=0$, $D^{\sigma}(\rho)>0$ for $\rho >0$, and $D^{\sigma}(\rho)/M \rho \in \Z[\rho]$.
\item[(b)] For any polynomial $\Delta(\rho) \coloneqq D^{\sigma'}(\rho) C^{\sigma,\sigma'}_{i}(\rho)/M \rho$ in $\cP$   (so $\sigma$ and $\sigma'$ are adjacent policies at state $i$) we have: $\Delta(\rho)=0 \iff C^{\sigma,\sigma'}_{i}(\rho)=0$, and $\Delta(\rho)$ and $C^{\sigma,\sigma'}_{i}(\rho)$ have the same sign.
\end{enumerate}
\end{lemma}
\begin{proof}
   (a) Since $D^{\sigma}(\rho)=\det \left( \rho M I - J^{\sigma}\right) = M^n \det \left( \rho I - L^{\sigma}\right)$, the facts that $D^{\sigma}(0)=0$ and $D^{\sigma}(\rho)>0$ for $\rho>0$ follow from $P^{\sigma}$ being a stochastic matrix, so $L^{\sigma}$ admits $0$ as an eigenvalue and all its other eigenvalues have non-positive real parts. Besides, since $D^{\sigma}(\rho)$ is a polynomial in $M\rho$ with a constant term equal to $0$, we have $D^{\sigma}(\rho)/M\rho \in \Z[\rho]$.

(b) This follows from the fact that $D^{\sigma}(\rho)/M\rho >0$ for $\rho>0$, as shown above. 
\end{proof}
We now show that the family $\cP$ is indeed a reduced family of comparison polynomials.
\begin{theorem}\label{th:cP is reduced family}
The family $\cP$ defined in~\eqref{eq:definition reduced family} is a reduced family of comparison polynomials.
\end{theorem}
\begin{proof}

The fact that the family $\cP$ is reduced follows readily from its definition, so we are left with proving that its elements are comparison polynomials, i.e., that they satisfy~\eqref{eq:def comp poly} and have integer coefficients. We first show that condition~\eqref{eq:def comp poly} holds. 
From item (c) of Lemma~\ref{lem:comp vector adjacent policies prop} and item (b) of Proposition~\ref{prop:resolvent properties - basic}, it follows that the sign of $C^{\sigma,\sigma'}_{i}(\rho)$ coincides with the sign of $(v^{\sigma}_{\rho} - v^{\sigma'}_{\rho})_i$. Besides, by Lemma~\ref{lem:zeros delta c}, the polynomial $D^{\sigma'}(\rho)C^{\sigma,\sigma'}_{i}(\rho)/M \rho$ has the same sign as $C^{\sigma,\sigma'}_{i}(\rho)$, so  condition~\eqref{eq:def comp poly} holds for the polynomials in $\cP$. 
It remains to show that the elements of $\cP$ have integer coefficients. 
Since $v^{\sigma'}_{\rho} = R^{\sigma'}(\rho)r^{\sigma'} = \left(\rho I - L^{\sigma'}\right)^{-1}r^{\sigma'}$,
using Cramer's formula for the inverse of a matrix and Laplace's cofactor expansion, it follows that 
\begin{equation}\label{ValueStoch}
(v_{\rho}^{\sigma'})_i= M \frac{\sum_{j=1}^{n} \cof_{ji}  (\rho M I - J^{\sigma'}) r_j^{\sigma'} }{D^{\sigma'}(\rho)}
\end{equation} 
for any state $i$. Therefore $D^{\sigma'}(\rho)C^{\sigma,\sigma'}_{i}(\rho)$ is equal to the following expression:
\begin{equation}\label{eq:Delta decomposition 0}
  D^{\sigma'}(\rho)r^{\sigma}_{i} + \sum_{h,j=1}^{n} M P^{\sigma}_{ih}\cof_{jh}(\rho M I - J^{\sigma'})r^{\sigma'}_{j} - M (1+\rho)\sum_{j=1}^{n} \cof_{ji}(\rho M I - J^{\sigma'})r^{\sigma'}_{j} \; .
\end{equation}
Every determinant and cofactor appearing in the above equation is a polynomial in $M\rho$ with integer coefficients because $J^{\sigma'}$ has integer coefficients. We are left with showing that its constant coefficient (as a polynomial in $\rho$) is zero.
The first term vanishes because $D^{\sigma'}(0)=0$. 
We get $\cof_{jh}(0 M I - J^{\sigma'}) = \cof_{jh}(-J^{\sigma'})$.
Overall, for $\rho=0$ in \eqref{eq:Delta decomposition 0} we want to show that $\displaystyle \sum_{j=1}^{n} r^{\sigma'}_{j}(\sum_{h=1}^{n} M P^{\sigma}_{ih}\cof_{jh}(-J^{\sigma'}) - M \cof_{ji}(-J^{\sigma'}))$ is equal to $0$. 
We claim that for any $h,h' \in [n]$ we have $\cof_{jh}(-J^{\sigma'})=\cof_{jh'}(-J^{\sigma'})$, i.e., the term $\cof_{jh}(-J^{\sigma'})$ does not depend on $h \in [n]$. 
If the rank of $J^{\sigma'}$ is at most $n-2$ then $\cof_{jh}(-J^{\sigma'})=0$ for any pair $(j,h) \in [n]$. 
Otherwise the rank of $J^{\sigma'}$ is $n-1$ and $\det(-J^{\sigma'})=0$. Now  $(\cof_{jh}(-J^{\sigma'}))_{h \in [n]}$ is the $j$-th column of the adjugate matrix $\adj(-J^{\sigma'}).$ From $(-J^{\sigma'}) \adj(-J^{\sigma'})= \det(-J^{\sigma'})I$ we get that the $j$-th column of the adjugate matrix belongs to $\ker(-J^{\sigma'})$. But if the rank of $J^{\sigma'}$ is $n-1$ then $\ker(-J^{\sigma'}) = \R e$ with $e=(1,...,1) \in \R^{n}$. Hence $(\cof_{jh}(-J^{\sigma'}))_{h \in [n]}$ is proportional to the vector $e$, i.e., $\cof_{jh}(-J^{\sigma'})$ does not depend on $h$. 
From there we get, using the stochasticity of $P^{\sigma}$, that $\sum_{h=1}^{n} M P^{\sigma}_{ih}\cof_{jh}(-J^{\sigma'}) - M \cof_{ji}(-J^{\sigma'})=0$ so that we have $\displaystyle \sum_{j=1}^{n} r^{\sigma'}_{j}(\sum_{h=1}^{n} M P^{\sigma}_{ih}\cof_{jh}(-J^{\sigma'}) - M \cof_{ji}(-J^{\sigma'})) = \sum_{j=1}^{n} r^{\sigma'}_{j} \cdot 0 = 0$. This shows that $D^{\sigma'}(\rho)C^{\sigma,\sigma'}(\rho)/M\rho \in Z[\rho]$ so that the elements of our reduced comparison family are indeed polynomials with integer coefficients. 

\end{proof}
We are now left with bounding the degree and height of the family $\cP$. We do so in the next section.
\subsubsection{Bounding the degree and height of the reduced family.}\label{sec:bounding height degree}
Bounding the degree of elements in $\cP$ is direct, but bounding their height requires a more careful analysis.
Our bounds are summarized in the following theorem.
\begin{theorem}\label{th:boundDeltaCoeff}
Let $\Delta (\rho) \coloneqq D^{\sigma'}(\rho) (C^{\sigma,\sigma'}(\rho))_i/M\rho$ be a polynomial in $\cP$. Then, $\Delta(\rho)$ is of the form $\sum_{k=0}^{n-1} b_k \rho^k$, where $b_k \in \Z$ and 
\begin{equation}\label{main_ineq}
|b_k| \leq 2 W M^{n-1} {n+1 \choose k+2}  \; .
\end{equation} 
\end{theorem}
To prove \Cref{th:boundDeltaCoeff}, we use variations on Tutte matrix tree theorem~\cite{Brualdi1991}, which provide combinatorial formulae for the coefficients of the characteristic polynomial of a Laplacian matrix.

We use the following graph theoretic notions.
A {\em spanning} subgraph of a digraph $\G$ is a subgraph of $\G$ which has the same vertex set as $\G$. A {\em directed tree} is a weakly connected digraph (i.e., its corresponding undirected graph is connected) in which one vertex, called {\em terminal}, has outdegree zero and the remaining vertices have outdegree one. A {\em directed forest} is a digraph all of whose weak components (i.e., maximal weakly connected subgraphs) are directed trees. The terminal vertices of these trees are the terminal vertices of the directed forest. When the original digraph is weighted, the weight of a forest is defined as the product of the weights of its edges. 

Henceforth, we denote by $\G^{\sigma}$ the digraph associated with the integer matrix $MP^{\sigma}$, by $c_k$ the sum of weights of all spanning directed forests of $\G^{\sigma}$ with $k$ edges, and by $q^{j,i}_k$ the sum of weights of all spanning directed forests of $\G^{\sigma}$ with $k$ edges in which state $i$ belongs to the tree rooted at $j$.
Note that $c_k$ and $q^{j,i}_k$ are integers because all the weights in $\G^{\sigma}$ are integers.
\begin{proposition}\label{Prop_det_bound}
The characteristic polynomial $D^{\sigma}(\rho)$ is of the form $\sum_{k=0}^n a_k \rho^k$ with $a_0=0$, where $a_k \in \Z$ and $|a_k| \leq  M^{n} {n \choose k}$ for any $k \in \{0,...,n\}$. 
\end{proposition}

\begin{proof}
The fact that $a_0=0$ follows from $D^{\sigma}(0)=0$ (\Cref{lem:zeros delta c}).
By~\cite[Proposition~2]{Chebotarev2002}, we have $\det(\rho M I - J^{\sigma})= \sum_{k=0}^n c_{n-k} M^k \rho^k$. Since the spanning directed forests with $k$ edges are the spanning directed forests with $n - k$ weak components, $c_k$ is also the sum of weights of all spanning directed forests with $n - k$ weak components. As mentioned above, with each spanning directed forest with $n - k$ weak components is associated a unique subset of $\{ 1, \dots , n \}$ with $n - k$ elements which is composed of the terminal vertices of its weak components (directed trees). Given such a set $I \subset \{ 1, \dots , n \}$, the sum of weights of all spanning directed forests with $k$ edges and terminal set $I$ is bounded by $M^k$. To see this, it is enough to notice that the weight of each of such spanning directed forests is bounded by a term of the expansion of 
\[
\prod_{i \notin I} \sum_{j \neq i} M P^{\sigma}_{ij} \; ,  
\]
and that no term needs to be used for more than one spanning directed forest. Indeed, each term of this expansion corresponds to a selection of an outgoing edge for each vertex that does not belong to $I$. Besides, the quantity above is less than or equal to $M^k$ due to the fact that $P^{\sigma}$ is (row) stochastic. Thus, we have
\begin{equation}\label{bound_inforests_k_edges}
c_k \leq \sum_{I \subseteq  \{ 1, \dots , n \}, |I|=n-k} \left( \prod_{i \notin I} \sum_{j \neq i} M P^{\sigma}_{ij} \right) \leq {n \choose k} M^k \; ,
\end{equation}
which completes the proof.

\end{proof}
\begin{proposition}\label{Prop_adj_bound}
For each state $i \in [n]$, the polynomial $D^{\sigma}(\rho) (v_{\rho}^{\sigma})_i$ is of the form $\sum_{k=0}^{n-1} a_k \rho^{k}$, where $a_k \in \Z$ and $|a_k| \leq W M^{n} {n \choose k+1}$ for any $k \in \{0,...,n-1\}$.
\end{proposition}
\begin{proof}
By~\cite[Proposition~3]{Chebotarev2002}, we have $ \cof_{ji} (\rho M I - J^{\sigma}) = \sum_{k=0}^{n-1} q^{j,i}_{n-k-1} M^{k} \rho^{k}$. Then, using~\eqref{ValueStoch} it follows that
\begin{equation}\label{Delta_v_i}
D^{\sigma}(\rho) (v_{\rho}^{\sigma})_i  =  M (\sum_{j=1}^{n} (\sum_{k=0}^{n-1} q^{j,i}_{n-k-1} M^{k} \rho^{k}) r_j^{\sigma}) 
 = \sum_{k=0}^{n-1} M^{k+1} \rho^{k} (\sum_{j=1}^{n} q^{j,i}_{n-k-1} r_j^{\sigma}) \; .
\end{equation}
Note that, independently of the choice of $i$, $\sum_{j=1}^{n} q^{j,i}_{n-k-1}$ is equal to the quantity $c_{n-k-1}$ defined above, as the sum of weights of all spanning directed forests of $\G^{\sigma}$ with $n-k-1$ edges, see~\cite[Proposition~1]{Chebotarev2002}. Then, by~\eqref{bound_inforests_k_edges} we have $\sum_{j=1}^{n} q^{j,i}_{n-k-1} \leq {n \choose n-k-1} M^{n-k-1} = {n \choose k+1} M^{n-k-1}$. Since $| r_j^{\sigma} | \leq W$ for each $j \in [n]$, this completes the proof.

\end{proof}
We are now ready to provide the proof of \Cref{th:boundDeltaCoeff}.
\begin{proof}[Proof of \Cref{th:boundDeltaCoeff}.]
Let $\Delta (\rho)= D^{\sigma'}(\rho)C^{\sigma,\sigma'}_{i}(\rho)/M\rho$ be an element of $\cP$. 
By~\eqref{eq:definition comparison vector}, we have
\begin{equation}\label{eq:Delta decomposition}
    M \rho \Delta(\rho) = D^{\sigma'}(\rho)r^{\sigma}_{i} + \sum_{h=1}^{n} P^{\sigma}_{ih} D^{\sigma'}(\rho) (v^{\sigma'}_{\rho})_h - D^{\sigma'}(\rho)(v^{\sigma'}_{\rho})_i - \rho D^{\sigma'}(\rho)(v^{\sigma'}_{\rho})_i \; .
\end{equation}
Since the right-hand side of~\eqref{eq:Delta decomposition} has degree at most $n$ by~\Cref{Prop_adj_bound}, $\Delta(\rho)$ is of degree at most $n-1$, i.e., it is of the form $\sum_{k=0}^{n-1} b_k \rho^{k}$.
By Propositions~\ref{Prop_det_bound} and~\ref{Prop_adj_bound}, the total contribution of the terms in~\eqref{eq:Delta decomposition} to the $k+1$-th power of $\rho$ (which corresponds to $b_k$) is bounded by
\[W M^n {n \choose k+1} + W M^{n} {n \choose k+2 } + W M^n {n \choose k+2} + W M^n {n \choose k+1} = 2 W M^n {n +1 \choose k+2}\]
where the last equality follows from Pascal's identity: ${n \choose k+1} + {n \choose k+2} = {n+1 \choose k+2}$. It follows that $M|b_k| \leq 2 W M^n {n +1 \choose k+2}$, which concludes the proof.

\end{proof}
\subsection{Proof of \Cref{th:reduced family bound - tightness}.}
\begin{proof}[Proof of \Cref{th:reduced family bound - tightness}.]
Consider the following MDP instance, assuming $M$, $W$ and $n$ satisfy the assumption of \Cref{th:reduced family bound - tightness}. There are $n$ states labeled $s_{0},...,s_{n-1}$ and two actions $a_0$ and $a_1$. The reward is $0$ in states $s_1,\ldots ,s_{n-2}$, and state $s_{n-1}$ is absorbing with reward $-1$. When choosing action $a_0$ in state $s_0$, the next state is $s_0$ with probability $1$ and the reward is $0$. When choosing action $a_1$ in state $s_0$, the reward is $2W-1$ and the next state is either $s_1$ with probability $1/M$ or $s_0$ with probability $1-1/M$. For $i = 1,\ldots ,n-2$, there is no action to choose in state $s_{i}$, where the transition toward $s_{i+1}$ has probability $1/M$ and the transition toward $s_0$ has probability $1-1/M$. 

The stationary deterministic policies are uniquely characterized by the action ($a_0$ or $a_1$) chosen at state $s_0$. Two entries of the value function when action $a_0$ is chosen are: $(v^{a_0}_{\rho})_{s_0} = 0$ and $(v^{a_0}_{\rho})_{s_1} = - 1/M^{n-2}\rho(1+\rho)^{n-2}$. 
Hence, $C^{a_1,a_0}_{s_0}(\rho) = (2W-1) - \frac{1}{M^{n-1}\rho (1+\rho)^{n-2}}$, and the Blackwell threshold $\rhobw$ is the unique positive root of the polynomial $p(\rho)=(2W-1)M^{n-1}\rho (1+\rho)^{n-2}-1$ (which is increasing on $(0,+\infty)$). It is straightforward to see that $p$ has degree $n-1$ and height $(2W-1)M^{n-1} {n-2 \choose \lfloor (n-2)/2 \rfloor}$.
To prove that $p$ is the minimal polynomial of $\rhobw$, it suffices to show that $p$ is irreducible. We proceed by showing that its reciprocal polynomial, $\tilde{p}(x) = x^{n-1}p(1/x) = (2W-1) M^{n-1} (x+1)^{n-2} - x^{n-1}$, is irreducible.  Recall that we have taken $2W-1$ prime such that $2W-1$ does not divide $M$. Therefore, the leading coefficient of $\tilde{p}$ (equal to $-1$) is not divisible by $2W-1$, every other coefficient is divisible by $2W-1$, and the constant coefficient of $\tilde{p}$ (equal to $(2W-1)M^{n-1}$) is not divisible by $(2W-1)^2$. Thus, we can apply Eisenstein's criterion with the prime $2W-1$ to conclude that $\tilde{p}$ (and therefore $p$) is irreducible. 

Finally, while the rewards in the MDP instance constructed above are between $-1$ and $2W-1$, we can always subtract $W-1$ from all rewards, without changing the Blackwell thresholds, to obtain an MDP with integer rewards between $-W$ and $W$. 
\end{proof}
\section{Bounds on the Blackwell thresholds.}\label{sec:separation blackwell}
We now present our bounds on the Blackwell thresholds, based on the analysis from the previous section.
\subsection{Root separation methods.}\label{sec:root separation}
\paragraph{Results based on Cauchy's bound.}
A classical and simple example of root separation is the following Cauchy's bound, which separates the positive roots of a polynomial from $0$.
\begin{lemma}[\cite{cauchy1837calcul}]\label{Lemma_Cauchy_Bound}
If $p(x)=\sum^{h}_{k=0} a_k x^k$ is a nonzero integer polynomial, any positive root $z>0$ of $p(x)$ satisfies $\frac{1}{2}\frac{1}{H(p)}\leq z.$
\end{lemma}
Cauchy's bound is used in~\cite{andersson2009complexity}, whose authors construct a polynomial in the variable $1-\alpha$ and separate its roots from $0$. Given a polynomial $p(x)$, to separate its positive roots from $1$ we can apply Cauchy's bound to the polynomial $q(y)$ defined as $y \mapsto p(1-y)$. However, doing so results in a binomial blow-up of the coefficients and therefore in the separation bounds. This is made precise in the next lemma.
\begin{lemma}\label{lem:cauchy from 1}
Let $p(x)=\sum^{h}_{k=0} a_k x^k$ be a nonzero integer polynomial. Then, $p(x)$ has no roots in $\left(1-\frac{1}{2 H(p){h+1\choose \lfloor (h+1)/2\rfloor}},1\right)$.
\end{lemma}
\begin{proof}
This is a direct application of Cauchy's bound to the polynomial $q(y) = p(1-y)$. By definition $q(y) = \sum_{i=0}^{h} b_i y_i$ with $b_i = (-1)^{i} \sum_{k=i}^{h} a_k {k\choose i} $. Since $\sum_{k=i}^h {k\choose i} = {h+1\choose i+1}$, we have $H(q) \leq H(p) {h+1\choose \lfloor (h+1)/2\rfloor}$. By \Cref{Lemma_Cauchy_Bound}, we conclude that $q(y)$ has no zeros in $\left(0, \frac{1}{2H(p) {h+1\choose \lfloor (h+1)/2\rfloor}}\right)$. 

\end{proof}
Equivalents of Lemma~\ref{lem:cauchy from 1} are used in~\cite{gaubert2026thresholds}. Other results, like a bound from Lagrange~\cite{Lagrange}, have been used in past works (e.g. in~\cite{mukherjee2025howard,gaubert2026thresholds}), but we find that the (arguably simpler) Cauchy's bound from Lemma~\ref{Lemma_Cauchy_Bound} yields the same final bounds on $\rhobw$ so we do not introduce the Lagrange bound in this paper.
\paragraph{Results based on Mahler measure.}
A recent work~\cite{gaubert2026thresholds} uses Mahler measure to provide bounds on $\alphabw$, albeit with an issue in the result on which it relies (see below). We provide a corrected version of this approach next.
\begin{definition}
If the polynomial $p(x)$ factorizes over the complex numbers as $p(x)=a \prod_{i=1}^{d}(x-z_{i})$, the Mahler measure $ \sfM(p) \in \R$ of $p$ is defined as $ \sfM(p) \coloneqq |a|  \prod_{i=1}^{d} \max \{1,|z_i|\}$. For $z \in \C$ algebraic, $ \sfM(z)$ is the Mahler measure of its minimal polynomial.
\end{definition}
Mahler measures can be used to separate algebraic numbers from $1$~\cite{mignotte1994algebraic,Dubickas1995}. We state the corrected version of Theorem 1 in~\cite{Dubickas1995} below.
\begin{theorem}\label{th:mahler}
For each $\epsilon>0$ there exists a constant $D_\epsilon \in \N$ such that, for any algebraic number $z$ of degree $d > D_\epsilon$ which is not a root of unity and satisfies $0 < \log( \sfM(z)) < d$, we have
\begin{equation}\label{eq:dubickas}
 |z -1 | > \exp\left(-(\pi/4 + \epsilon) \sqrt{d \log d  \log  \sfM(z)} \right) \; .
 \end{equation}
 \end{theorem}
 Note that Theorem~\ref{th:mahler} does require the ``small Mahler measure'' condition $0 < \log( \sfM(z)) < d$, which is omitted in the statement of Theorem 1 in~\cite{Dubickas1995} but acknowledged by the same author in the later paper~\cite{dubickas1998algebraic}. Without this condition, the inequality~\eqref{eq:dubickas} cannot hold (as shown in the introduction in~\cite{mignotte1994algebraic}). 
To handle the cases not covered under the assumptions of Theorem~\ref{th:mahler} ($d \leq D_{\epsilon}$ or $d \leq \log( \sfM(z))$), we can resort to Liouville's inequality, which is always true: $|z-1| \geq 2^{1-d}  \sfM(z)^{-1}$, where $d$ is the degree of $z$. This yields the following bound.
\begin{theorem}\label{th:mahler sep - generic}
Let $\theta \coloneqq 1 + \log(2)$. For each $\epsilon>0$ there exists a constant $c_\epsilon  \geq 1$ such that, for any algebraic number $z$ of degree $d$ which is not a root of unity, we have
\begin{equation}\label{eq:dubickas-rev}
|z-1| \geq \min \left\{ \exp\left(-\left(\frac{\pi}{4}+\epsilon \right) \sqrt{d\log(d)\log(\sfM(z))}\right), \frac{1}{c_{\epsilon} (\sfM(z))^{\theta}} \right\}\; .
\end{equation}
\end{theorem}
\begin{proof}
Given $\epsilon>0$, let $D_{\epsilon}\in \N$ be the constant provided by Theorem~\ref{th:mahler}. Let $z>0$ be an algebraic number of degree $d$ which is not a root of unit. We consider three cases.
In the first case, assume $d> D_{\epsilon}$ and $\log( \sfM(z)) < d$. Then, we can apply Theorem~\ref{th:mahler} directly to obtain $|z-1| > \exp\left(-(\frac{\pi}{4}+\epsilon) \sqrt{d\log(d)\log(\sfM(z))}\right)$. In the second case, assume $d > D_{\epsilon}$ but $\log( \sfM(z)) \geq d$. Then, Liouville's inequality gives $|z-1| \geq 2^{1-d}\sfM(z)^{-1} = 2 \exp(-\log( \sfM(z)) - d \log(2)) \geq 2 \exp(-(1+\log(2))\log( \sfM(z))) = 2  \sfM(z)^{-\theta}$. 
In the third case, assume $d \leq D_{\epsilon}$. Then, Liouville's inequality gives $|z-1| \geq 2^{1-d} \sfM(z)^{-1} \geq 2^{1-D_{\epsilon}} \sfM(z)^{-\theta}$ (since $\theta>1$). Taking $c_{\epsilon} \coloneqq \max \{1,2^{D_{\epsilon}-1}\}$ concludes the proof.
 
\end{proof}
To apply \Cref{th:mahler sep - generic}, we need an upper bound on $ \sfM(z)$ when $z$ is a root of a polynomial $p(x) = \sum^{d}_{k=0} a_k x^k$. We will use the bound $\sfM(z) \leq \sqrt{\sum^{d}_{k=0} a^2_k}$, which follows from the fact that $\sfM(z) \leq \sfM(p)$ since the minimal polynomial of $z$ divides $p$ (given that $p(z)=0$), and a classical bound from Landau~\cite{Landau1905}, which states that for any complex polynomial $q(x) =\sum^{d}_{k=0} b_k x^k$ we have $ \sfM(q) \leq \sqrt{\sum^{d}_{k=0} |b_k|^2}$. We are now ready to bound the Blackwell threshold.
\subsection{Lower bounds on $\rhobw$.}
We have established in \Cref{th:reduced family bound} the existence of a reduced family $\cP$ of controlled degree and height, while \Cref{th:reduced family root bound alphabw} shows that $\rhobw$ is bounded below by the minimum positive root across all polynomials in $\cP$. 
We can now apply the separation results from the previous section to the polynomials in the reduced family constructed in \Cref{th:reduced family bound}. The next theorem uses Cauchy's bound.
\begin{theorem}\label{Th:boundBlackwellCauchy}
The Blackwell threshold $\rhobw$ of a two-player perfect-information stochastic game satisfying Assumption \ref{assumption} verifies
  \[
\frac{1}{4 W M^{n-1} {n+1 \choose \lfloor (n+1)/2 \rfloor}}  \leq  \rhobw  \; .
  \]
\end{theorem} 
\begin{proof}
From \Cref{prop:reduction games to mdps}, it is sufficient to prove this statement for the set of MDPs satisfying Assumption~\ref{assumption}.
From Theorem~\ref{th:reduced family root bound alphabw}, we can lower bound $\rhobw$ by the smallest positive root  across all polynomials in a reduced family. By Theorem~\ref{th:reduced family bound}, there exists a reduced family of polynomials with degree at most $n-1$ and height at most $2 W M^{n-1} {n+1 \choose \lfloor (n+1)/2 \rfloor}$. We conclude by a direct application of Cauchy's bound as in \Cref{Lemma_Cauchy_Bound}.

\end{proof}
The next bound uses the separation result based on Mahler measures (\Cref{th:mahler sep - generic}).
\begin{theorem}\label{th-dubickas} 
For each $\epsilon > 0$ there exists a constant $c_{\epsilon} \geq 1$ such that the Blackwell threshold $\rhobw$ of any two-player perfect-information stochastic game satisfying Assumption~\ref{assumption} verifies
\[\min \left\{ \exp\left(-\left(\frac{\pi}{4}+\epsilon \right) \sqrt{n\log(n)\log(B)}\right), \frac{1}{c_{\epsilon} B^{\theta}} \right\} \leq \rhobw \; ,\]
where $\displaystyle\theta \coloneqq 1 + \log(2)$ and $B \coloneqq \frac{1}{2}W M^{n-1}3^{n+1}$.
\end{theorem}
\begin{proof}
As in the proof of \Cref{Th:boundBlackwellCauchy}, we can lower bound $\rhobw$ by the smallest positive root across all polynomials in the reduced family $\cP$, so let 
 $\Delta(\rho)$ be any of its elements. \Cref{th:mahler sep - generic} can be used to separate a root of a polynomial from $1$, so we consider the polynomial $p(\beta) \coloneqq \Delta(\beta -1)$ to separate the positive roots of $\Delta(\rho)$ from $0$. By \Cref{th:boundDeltaCoeff} we know that $\Delta(\rho)$ is of the form $\sum_{k=0}^{n-1} b_k \rho^{k}$, and so $p(\beta) = \sum_{i=0}^{n-1} a_i \beta^i$,
where $a_i = \sum_{k=i}^{n-1} (-1)^{k-i} b_k {k\choose i}$. By Landau's bound we get 
\[ \sfM(p) \leq \sqrt{\sum_{i=0}^{n-1} |a_i|^2} \leq \sum_{i=0}^{n-1} | a_i| \leq \sum_{k=0}^{n-1} |b_k| \sum_{i=0}^{k} {k\choose i} = \sum_{k=0}^{n-1} |b_k|2^k.\]
Then, by~\eqref{main_ineq} we have $\sum_{k=0}^{n-1} |b_k|2^k \leq 2 W M^{n-1} \sum_{i=0}^{n-1} \sum_{k=i}^{n-1} {n+1 \choose k +2} {k\choose i} = 2 W M^{n-1} \frac{3^{n+1} - 2n - 3}{4} \leq \frac{1}{2}W M^{n-1}3^{n+1} $ since $\sum_{k=0}^{n} {n \choose k} 2^{k} = 3^{n}$. 
The theorem now follows from the fact that any root $z \in (1,2)$ of $p(\beta)$ of degree $d$ satisfies the inequality~\eqref{eq:dubickas-rev}, along with the facts that $d\leq n$ and $\sfM(z) \leq \sfM(p) \leq \frac{1}{2}W M^{n-1}3^{n+1}$. 

\end{proof}
Unfortunately, the bound on $\rhobw$ obtained in Theorem~\ref{Th:boundBlackwellCauchy} always dominates the one from Theorem~\ref{th-dubickas}. Indeed the bound based on Cauchy is always larger than $(4WM^{n-1}2^{n+1})^{-1}$ and this last quantity is itself larger than $(WM^{n-1}3^{n+1})^{-\theta} = B^{-\theta}$ (for $\theta = 1 + \log(2)>1$, $n \geq 2$, and $B$ as  defined in Theorem~\ref{th-dubickas}), which is always larger than the minimum on the right hand side in the bound from Theorem~\ref{th-dubickas}.
This happens because the upper bound $B$ on the Mahler measure $ \sfM(\rhobw)$ in Theorem~\ref{th-dubickas} is very large (exponential in the number of states). Yet, bounds based on Mahler measures will prove useful in the case of deterministic instances where we are better able to control $ \sfM(\rhobw)$, as will be explained in Sections \ref{sec:deter d sens} and~\ref{sec:discussion}.
\subsection{Upper bounds on $\rhobw$.}\label{sec:upper bound rho}
We now present some examples of MDP instances where $\rhobw$ can be exponentially (or at least superpolynomially) close to $0$.
\subsubsection{Stochastic instances.}\label{sec:lower bound sto}
The following theorem provides an upper bound on $\rhobw$ for the case of general MDP instances, including the case of randomized transitions ($M \geq 2$).
\begin{theorem}\label{th:lower bound stochastic}
For each $n \geq 2$ and $M,W \in \N$ there exists an MDP instance satisfying Assumption~\ref{assumption} such that
\[\rhobw \leq \frac{1}{(2W-1) M^{n-1}} \; .\]
\end{theorem}
\begin{proof}
Consider the MDP instance from the proof of Theorem~\ref{th:reduced family bound - tightness}. As shown in that proof, $\rhobw$ is the unique positive solution to $(2W-1)M^{n-1} x (1+x)^{n-2}=1$. Since $(1+\rhobw)^{n-2} \geq 1$, we have $\rhobw \leq \frac{1}{(2W-1) M^{n-1}}$.

\end{proof}
To the best of our knowledge, only two prior works provide results on upper bounds for $\rhobw$ in the stochastic case. Theorem 1 of~\cite{boros2013discounted} shows that we may have $\rhobw \leq O\left(n^{1/12}/2^{n/3}\right)$ (showing the same inequality for $\rho_{-1}$) for an MDP instance with non-deterministic transitions out of only one node, $O(n)$ states and $W=O(n^{20})$. In~\cite[Theorem 2]{mukherjee2025efficient}, the authors construct an MDP instance satisfying Assumption~\ref{assumption} where $\rhobw < \frac{M}{\Omega(2^{n/3})}$ for $M \in \N$ and $W=1$ (see Appendix~\ref{app:comparison mukherjee} for more details; in fact the authors show this for constant integer values of $M$ but their proof extends easily to the case of $M \in \N$). Therefore, our bound of \Cref{th:lower bound stochastic} improves the ones in~\cite{boros2013discounted,mukherjee2025efficient} both in terms of the exponents in $M$ and the factor $W$, and it exactly matches the bound of \Cref{Th:boundBlackwellCauchy} regarding the term $M^{n-1}$. Achieving the missing term of order ${n \choose \lfloor n/2 \rfloor}$ in the denominator is an important open problem.
\subsubsection{Deterministic instances.} 
It is worth noting that the bound obtained in the previous section for {\em stochastic} instances relies crucially on the randomness of the transitions: the term $M^{n-1}$ in the denominator appears from a chain of length $n$ and a sequence of transition probabilities $1/M$ along that chain. Hence, the bound of \Cref{th:lower bound stochastic} is nearly vacuous when $M=1$. Therefore, bounds for deterministic instances require completely different ideas than for stochastic instances. Our proof strategy is to construct $\alphabw$ very close to $1$ and then use the fact that $\rhobw \leq 2 (1-\alphabw)$ for $\alphabw \geq 1/2$ to obtain a bound on $\rhobw$.
Constructing such $\alphabw$ is tightly connected to constructing polynomials with roots close to $1$, as has been observed repeatedly in the past, e.g.~\cite{asadi2024deterministic} or Section 3 in~\cite{grand2023reducing}. We provide a concise proof for completeness.
\begin{lemma}\label{lem:reduction poly to alphabw}
Let $p \in \Z[x]$ be a nonzero polynomial with at least one root in $(0,1)$. Then, there exists a deterministic MDP instance with $\deg(p)+1$ states, $2$ actions, maximum rewards bounded in absolute value by the height $H(p)$ of $p$, such that the Blackwell threshold $\alphabw$ is the largest root of $p$ in $(0,1)$.
\end{lemma}
\begin{proof}
Assume $p(x) = \sum_{i=0}^{h} c_i x^i$ for some $h \in \N$ and $c_i \in \Z$ for $i \in \{0,...,h\}$. We construct the MDP instance as follows. There are $h+1$ states labeled $s_0,\ldots,s_h$, two actions $a_0$ and $a_1$ to choose in state $s_0$, and no choices in the other states (the only decision is in state $s_0$). If action $a_0$ is chosen in $s_0$ the next state is $s_{0}$ with probability $1$ and the reward is $0$, while if action $a_1$ is chosen in $s_0$ the next state is $s_1$ with probability $1$ and the reward is $c_0$.  
For $i \in \{1,...,h-1\}$, the transition from $s_{i}$ to $s_{i+1}$ is deterministic with reward $c_i$. From $s_h$, the MDP transitions deterministically back to $s_0$ with reward $c_h$. A policy is parametrized by the action chosen at $s_0$. 
The return for $a_0$ is $0$, while the return for $a_1$ is $\left(\sum_{i=0}^{h} c_i \alpha^{i+1}\right)/(1-\alpha^{h+1})=\alpha p(\alpha)/(1-\alpha^{h+1})$. Since the signs of the two polynomials $p(\alpha)$ and $\alpha p(\alpha)/(1-\alpha^{h+1})$ coincide for $\alpha \in (0,1)$, the Blackwell discount factor $\alphabw$ is the largest positive root of $p$ strictly smaller than $1$. 
\end{proof}

From Lemma~\ref{lem:reduction poly to alphabw}, we are left to construct polynomials with integer coefficients with roots very close to $1$. The next theorem is an extension of Theorem 5 from~\cite{borwein1997polynomials}.
\begin{theorem}\label{th:borwein pinner - coefficient H}
There exists $c>0$ such that for any $H \in \N$, there exists infinitely many $N \in \N$ such that there is a polynomial of degree at most $N$, with coefficients in $\{-H,...,H\}$, and with a real root $\alpha_{N,H} \in (0,1)$ satisfying
\[1-\alpha_{N,H} < \frac{1}{H} \exp \left( - c (N \log(N))^{1/3}\right).\]
\end{theorem}
For conciseness the proof is presented in Appendix~\ref{app:lower bound deterministic}.
Our hard-instance upper bound on $\rhobw$ for deterministic instances now follows easily by an application of Lemma~\ref{lem:reduction poly to alphabw}.
\begin{theorem}\label{th:lower bound deterministic}
There exists $c>0$ such that for any $W \in \N$, there exists infinitely many $n \in \N$ such that there is a deterministic MDP instance with $n$ states and rewards bounded by $W$ in absolute value, such that the Blackwell threshold of this MDP instance satisfies
\[\rhobw < \frac{1}{W} \exp \left( - c (n \log(n))^{1/3}\right).\]
\end{theorem}
In terms of the dependency on $n$ of $\log(\rhobw^{-1})$ for fixed $W \in \N$,  this leaves open a gap between $O(n)$ implied by Theorem~\ref{Th:boundBlackwellCauchy} and $\Omega((n\log(n))^{1/3})$ from the instance in Theorem~\ref{th:lower bound deterministic}. Closing this gap likely requires closing the gap on the multiplicity of real polynomials at $1$, which appears difficult. We note that Theorem 4.2 in~\cite{asadi2024deterministic} yields deterministic two-player perfect-information SGs with $\log(\rhobw^{-1}) =\Omega((\log(n))^2)$ and $\log(\rhobw^{-1}) =\Omega(n^{1/5}\sqrt{\log(n)})$ (with $W=2$), which both yield weaker bounds than Theorem~\ref{th:lower bound deterministic} for the asymptotic behavior of $\log(\rhobw^{-1}).$
\section{Bounds for the $d$-sensitive thresholds.}\label{sec:bound d sensitive}
We now bound the $\sens$-sensitive thresholds of games satisfying Assumption~\ref{assumption}. Following the reduction of Proposition~\ref{prop:reduction games to mdps}, we will do so for MDPs satisfying this assumption.
The $d$-sensitive optimality criterion was introduced in Definition~\ref{def:d sensitive optimality} in terns of limits of functions of the discount factor $\alpha$, but it is well-known that they can be equivalently expressed in terms of the interest rate $\rho$ as follows.
\begin{lemma}[\cite{puterman2014markov}, Equation (10.1.6)]
Let $d \in \{-1,0,1,\ldots\}$ and $\sigma\opt$ be a stationary policy. Then, $\sigma\opt$ is $d$-sensitive optimal if and only if for any other policy $\sigma$ we have
\begin{equation}\label{eq:d sensitive - rho}
\liminf_{\rho \to 0^+} \rho^{-d}\left( v^{\sigma}_{\rho} - v^{\sigma\opt}_{\rho}\right) \geq 0\; .
\end{equation}
\end{lemma}
\subsection{A sufficient criterion for $d$-sensitive optimality.}
The rest of this section extensively uses the {\em Laurent series expansion} of the value functions in terms of the interest rate $\rho$, and what \eqref{eq:d sensitive - rho} implies for the coefficients in this expansion. In order to better understand the proofs in the rest of this section, it is worth focusing first on the Laurent series expansions of value functions. For a given stationary policy $\sigma$, there exists a sequence $\left(V_{[k]}^{\sigma}\right)_{k \geq -1}$ of vectors of $\R^{n}$ (for instance $V^{\sigma}_{[-1]}$ is the gain of $\sigma$) such that
\begin{equation}\label{eq:laurent series v}
v^{\sigma}_{\rho}  = \sum_{k=-1}^\infty \rho^k V^\sigma_{[k]}
\end{equation}
for $\rho>0$ sufficiently small (e.g. see~\cite{Veinott1969} for more details). 
Since the $d$-sensitive optimality requires analyzing the limits of expressions of the form $\rho^{-d}v^{\sigma}_{\rho}$ as $\rho \to 0^+$, we provide a higher-level analysis in the next lemma.
\begin{lemma}\label{lem:limit sign of Laurent expansion}
Assume $f:\R_{+}^{*} \rightarrow \R$ admits a Laurent series expansion close to $0$, i.e., there exist some scalars $\left(f_k\right)_{k \geq -1}$ such that $f(\rho) = \sum_{k=-1}^{+\infty} f_k \rho^k$ for small $\rho>0$. Given $d \in \{-1,0,1, \dots\}$, if $f_{-1} = \ldots = f_{d}=0$ then $ \displaystyle \liminf_{\rho \to 0^+} \rho^{-d}f(\rho)=0$. Otherwise, $\displaystyle \liminf_{\rho \to 0^+} \rho^{-d}f(\rho)$ has the same sign as the first nonzero coefficient among $f_{-1},\ldots ,f_{d}$.
\end{lemma}
\begin{proof}
If $f_{-1} = \ldots = f_{d}=0$, then $f(\rho) = O(\rho^{d+1})$ (and so $\rho^{-d}f(\rho) = O(\rho)$) for $\rho>0$ small, which implies that $\displaystyle \liminf_{\rho \to 0^+} \rho^{-d}f(\rho)=0$. Otherwise, let $d'$ be the index of the first nonzero coefficient among $f_{-1},\ldots,f_{d}$. Then, we have $\rho^{-d}f(\rho) = \rho^{d'-d}(f_{d'} + o(1))$, and therefore the sign of $\displaystyle \liminf_{\rho \to 0^+} \rho^{-d}f(\rho)$ is the same as the sign of $f_{d'}$.
 
\end{proof}
For conciseness we use the following lexicographic ordering for matrices.
\begin{definition}[Lexicographic ordering]
Given two real matrices $A$ and $B$ of the same size, 
we write $A \succeq B$ if, in each row of $A-B$, either all entries are $0$ or the first nonzero entry is positive. We write $A \succ B$ if $A \succeq B$ and $A \neq B$. 
\end{definition}
In terms of the lexicographic ordering, Lemma~\ref{lem:limit sign of Laurent expansion} says that $(f_{-1},\ldots,f_{d}) \succeq 0$ is equivalent to $\displaystyle \liminf_{\rho \to 0^+} \rho^{-d}f(\rho) \geq 0$. We will use this equivalence repeatedly in the rest of this section.

Let us now consider the case of the difference of two value functions. By~\eqref{eq:laurent series v} we know that
\begin{equation}\label{eq:d sens vs. laurent}
v^{\sigma}_{\rho} - v^{\sigma'}_{\rho} = \sum_{k=-1}^\infty \rho^{k} \left(V^\sigma_{[k]}- V^{\sigma'}_{[k]}\right)
\end{equation}
for $\rho>0$ small enough. A direct application of Lemma~\ref{lem:limit sign of Laurent expansion} gives the next well-known theorem, where we denote by $V_{[-1,d]}^{\sigma}$ the matrix with columns $V_{[-1]}^{\sigma},\dots ,V_{[d]}^{\sigma}$. 
\begin{theorem}\label{th:d sens lex ordering}
A policy $\sigma\opt$ is $d$-sensitive optimal if, and only if, $V_{[-1,d]}^{\sigma} \succeq V_{[-1,d]}^{\sigma\opt}$ for any other policy $\sigma$.
\end{theorem}
The previous theorem was first proved in~\cite{Veinott1969}; we refer to Theorem 10.2.3 in~\cite{puterman2014markov} for a more modern presentation. 
It shows that $d$-sensitive optimal policies lexicographically minimize the Laurent coefficients of the value functions, from order $-1$ to order $d$. 
We will show that a sufficient condition for the $d$-sensitive optimality of a policy $\sigma\opt$ can be given by considering the Laurent series expansion of the comparison vectors $C^{\sigma,\sigma\opt}(\rho)$ for policies $\sigma$ that are adjacent to $\sigma\opt$  (akin to Proposition~\ref{PropVeinott} for discount optimality).  This is due to the fact that the performance difference lemma~\eqref{eq:performance difference} allows us to relate the Laurent series of $v^{\sigma}_{\rho} - v^{\sigma'}_{\rho}$ with the Laurent series of $R^{\sigma}(\rho)$ and $C^{\sigma,\sigma'}(\rho)$. 
Computing the Laurent series expansion of $C^{\sigma,\sigma'}(\rho)$ is straightforward, given~\eqref{eq:laurent series v} and the definition of comparison vectors~\eqref{eq:definition comparison vector}. More precisely, we have 
\begin{equation}\label{eq:laurent series C}
C^{\sigma, \sigma'} (\rho ) = \sum_{k=-1}^\infty \rho^k C_{[k]}^{\sigma , \sigma'}
\end{equation}
where $C_{[k]}^{\sigma , \sigma'} \coloneqq 1_{ \{k =0\} }r^{\sigma} + L^{\sigma} V_{[k]}^{\sigma'} - V_{[k-1]}^{\sigma'}$, $V_{[-2]}^{\sigma\opt} \coloneqq 0$, and $1_{\{k=0\}}$ equal to $1$ if $k=0$ and $0$ otherwise.

The following lemma describes the asymptotic behavior of the resolvent $R^{\sigma}(\rho)$ close to $0$.
\begin{lemma}\label{lem:resolvent properties}
Let $\sigma$ be a stationary policy. For $i \in [n]$, we define $\epsilon^{\sigma}_{i} \in \{0,1\}$ as follows: $\epsilon^{\sigma}_{i}\coloneqq 1$ if state $i$ is recurrent under $P^{\sigma}$, and $\epsilon^{\sigma}_{i}\coloneqq 0$ otherwise. Then, for each $h \in [n]$, either $R_{hi}^{\sigma}$ is identically zero or there exists $c_{hi}>0$ such that $R_{hi}^{\sigma}(\rho) = c_{hi} \rho^{-\epsilon^{\sigma}_i} + O(\rho^{1-\epsilon^{\sigma}_i})$ for $\rho>0$.
\end{lemma}
\begin{proof}
If $i$ is transient, we have 
$\lim_{\rho \to 0^+} R^{\sigma}_{hi}(\rho) =\lim_{\rho \to 0^+} \left(\sum_{k=0}^{\infty} (1+\rho)^{-(k+1)} \left((P^{\sigma})^{k}\right)_{hi}\right) = \sum_{k=0}^{\infty} \left((P^{\sigma})^{k}\right)_{hi}<+\infty$.  
Hence, if $R^{\sigma}_{hi}(\rho)$ is not identically zero, this limit is positive and the first Laurent coefficient has index $0$ and is positive. 
If $i$ is recurrent, we use~\cite[Theorem 3]{veinott}, which shows that $\rho R^{\sigma}(\rho) = (P^{\sigma})^* + O(\rho)$, where $(P^{\sigma})^*$ is the limiting matrix of $P^{\sigma}$ (i.e., the Cesaro mean of $P^{\sigma}$). If $i$ is not reachable from $h$ under $P^\sigma$, then $R^{\sigma}_{hi}(\rho)$ is identically zero. Otherwise, we have $(P^{\sigma})^*_{hi}>0$, which completes the proof of the lemma.
 
\end{proof}

In the particular case of adjacent policies, the following lemma describes some properties of $C_{[k]}^{\sigma,\sigma'}$ that mirror the ones for comparison vectors from Lemma~\ref{lem:comp vector adjacent policies prop}.
\begin{lemma}\label{lem:comp matrix adjacent policies prop}
Let $\sigma$ and $\sigma'$ be adjacent policies and $i \in [n]$ the state where they differ. Then:
\begin{itemize}
\item[(a)] For $j \neq i$, we have $(C^{\sigma,\sigma'}_{[k]})_j =0 $ for any $k \geq -1$.
\item[(b)] For any $k \geq -1$, the $i$-th entry of $C^{\sigma,\sigma'}_{[k]}$ only depends on $\sigma'$ and $\sigma(i)$, not on $\sigma(j)$ for $j \neq i$.
\end{itemize}
\end{lemma}
\begin{proof}
Both statements rely on the closed-form expression of $C_{[k]}^{\sigma , \sigma'}$ given right after~\eqref{eq:laurent series C}: for any $h \in [n]$ we have $(C_{[k]}^{\sigma , \sigma'})_h= 1_{ \{k =0\} }r^{\sigma}_{h} + (L^{\sigma} V_{[k]}^{\sigma'})_{h} - (V_{[k-1]}^{\sigma'})_h$.
 
(a) Since $\sigma$ and $\sigma'$ coincide for $j \neq i$ we get $(C_{[k]}^{\sigma , \sigma'})_j= 1_{ \{k =0\} }r^{\sigma}_{j} + (L^{\sigma} V_{[k]}^{\sigma'})_{j} - (V_{[k-1]}^{\sigma'})_j$, and \cite[Theorem 8.2.8]{puterman2014markov} shows that $(C_{[k]}^{\sigma , \sigma'})_j=0$. We note that a more direct proof uses the fact that $C^{\sigma,\sigma'}_{j}(\rho)=0$ for any $\rho>0$, and that a Laurent series is identically equal to $0$ if and only if all its coefficients are zero.

(b) The $i$-th entry of $C^{\sigma,\sigma'}_{[k]} $ is
$1_{ \{k =0\} }r^{\sigma}_{i} + (L^{\sigma} V_{[k]}^{\sigma'})_i - (V_{[k-1]}^{\sigma'})_i$. The claim now follows from the fact that $r^{\sigma}_{i}$ and the $i$th row of $L^{\sigma}$ only depend on $\sigma(i)$.  
\end{proof}
In what follows, given $d \in \{0,1,\ldots \}$, we denote by $C_{[-1,d]}^{\sigma,\sigma'}$ the matrix with columns $C_{[-1]}^{\sigma,\sigma'},\dots ,C_{[d]}^{\sigma,\sigma'}$. Note that item (a) of \Cref{lem:comp matrix adjacent policies prop} shows that, if $\sigma$ and $\sigma'$ are adjacent policies that differ at state $i$, then $(C_{[-1,d]}^{\sigma,\sigma'})_j = 0$ for any $d \in \{0,1,\ldots \}$ and $j \neq i$. 

We are now ready to give a sufficient condition for $d$-sensitive optimality.

\begin{proposition}\label{Prop:CharactDSen}
Let $\sigma\opt$ be a stationary policy and $d \in \{-1,0,1, \dots\}$. If $C_{[-1,d+1]}^{\sigma, \sigma^\ast} \succeq 0$ for every stationary policy $\sigma$ adjacent to $\sigma\opt$, then $\sigma\opt$ is $d$-sensitive optimal.
\end{proposition}
\begin{proof}
In the first place we show that, under the given assumption, we actually have $C_{[-1,d+1]}^{\sigma, \sigma\opt} \succeq 0$ for any policy $\sigma$ (not just for policies adjacent to $\sigma\opt$). To see this, let $\sigma $ be any policy and $i \in [n]$ be any state. Let $\sigma'$ be the policy defined by $\sigma'(i)\coloneq\sigma(i)$ and $\sigma'(j)\coloneq\sigma\opt(j)$ for $j \neq i$. Then, since $\sigma'$ is adjacent to $\sigma\opt$, by assumption we have $C_{[-1,d+1]}^{\sigma', \sigma\opt} \succeq 0$, and so $(C_{[-1,d+1]}^{\sigma', \sigma\opt})_i \succeq 0$. Moreover, by item (b) of Lemma~\ref{lem:comp matrix adjacent policies prop} and the fact that $\sigma'(i)=\sigma(i)$, we have $(C^{\sigma,\sigma\opt}_{[k]})_i = (C^{\sigma',\sigma\opt}_{[k]})_i$ for any $k$. We conclude that $(C_{[-1,d+1]}^{\sigma, \sigma\opt})_i \succeq 0$. Since this holds for any state $i \in [n]$, it proves our claim. 

Now we prove that $\sigma\opt$ is $d$-sensitive optimal. Let $\sigma$ be any stationary policy and $i \in [n]$ be any state. By~\eqref{eq:performance difference} we have $\displaystyle (v^{\sigma}_{\rho} - v^{\sigma\opt}_{\rho})_i = \sum_{j=1}^{n} R^{\sigma}_{ij}(\rho) C^{\sigma,\sigma\opt}_{j}(\rho)$, 
and therefore
\begin{equation}\label{eq:perf diff S1 S2}
 \frac{(v^{\sigma}_{\rho} - v^{\sigma\opt}_{\rho})_i}{\rho^{d}}   = \sum_{j \in [n], (C_{[-1,d+1]}^{\sigma, \sigma\opt})_j = 0} \frac{R^{\sigma}_{ij}(\rho) C^{\sigma,\sigma\opt}_{j}(\rho)}{\rho^{d}} + \sum_{j \in [n], (C_{[-1,d+1]}^{\sigma, \sigma\opt})_j \neq 0}  \frac{R^{\sigma}_{ij}(\rho) C^{\sigma,\sigma\opt}_{j}(\rho)}{\rho^{d}} \; .
\end{equation}
We now analyze the terms on the right-hand side of~\eqref{eq:perf diff S1 S2}. 
\begin{itemize}
\item If $(C_{[-1,d+1]}^{\sigma, \sigma\opt})_j = 0$, the terms of order $-1,0,...,d+1$ of the Laurent expansion of $C^{\sigma,\sigma\opt}_{j}(\rho)$ are null, so we have $C^{\sigma,\sigma\opt}_{j}(\rho) = O\left(\rho^{d+2}\right)$. Besides, by \Cref{lem:resolvent properties} we have $R^{\sigma}_{ij}(\rho) = O(\rho^{-\epsilon^{\sigma}_{j}})$, where $\epsilon^{\sigma}_{j} \in \{0,1\}$. It follows that $R^{\sigma}_{ij}(\rho) C^{\sigma,\sigma\opt}_{j}(\rho) = O(\rho^{d+1})$, and so $\liminf_{\rho \to 0^+} \rho^{-d} R^{\sigma}_{ij}(\rho) C^{\sigma,\sigma\opt}_{j}(\rho) = 0$.
\item If $(C_{[-1,d+1]}^{\sigma, \sigma\opt})_j \neq 0$, by the first part of the proof we must have $(C_{[-1,d+1]}^{\sigma, \sigma\opt})_j \succ 0$. Then, if $k'$ is the smallest $k \in \{-1,...,d+1 \}$ such that $(C^{\sigma,\sigma\opt}_{[k]})_j \neq 0$, we have $(C^{\sigma,\sigma\opt}_{[k']})_j>0$ and $
C^{\sigma,\sigma\opt}_{j}(\rho) = \rho^{k'} \left((C^{\sigma,\sigma\opt}_{[k']})_j  + o(1)\right)$. It follows that $C^{\sigma,\sigma\opt}_{j}(\rho)$ is positive for $\rho$ sufficiently small, and therefore this is still the case for $\rho^{-d}R^{\sigma}_{ij}(\rho) C^{\sigma,\sigma\opt}_{j}(\rho)$ because $R^{\sigma}(\rho)\geq 0$ by \Cref{prop:resolvent properties - basic}.  
Thus, we have $\liminf_{\rho \to 0^+} \rho^{-d}R^{\sigma}_{ij}(\rho) C^{\sigma,\sigma\opt}_{j}(\rho) \geq 0$. 
\end{itemize}
It follows that every term on the right-hand side of~\eqref{eq:perf diff S1 S2} has a non-negative $\liminf$ as $\rho \to 0^+$, and therefore $\liminf_{\rho \to 0^+} \rho^{-d}\left(v^{\sigma}_{\rho} - v^{\sigma\opt}_{\rho}\right)_i \geq 0$. Since this holds for any stationary policy $\sigma$ and any $i\in [n]$, we conclude that $\sigma\opt$ is $d$-sensitive optimal.

\end{proof}
The previous proposition improves \cite[Theorem 4]{Veinott1969} by considering only adjacent policies.

\subsection{Polynomials admitting $\rhod$ as a root.}

Let $\sigma$ and $\sigma'$ be adjacent policies at state $i \in [n]$. We denote by $T_{i}^{\sigma,\sigma'} \subset [n]$ the set of states reachable by $P^{\sigma'}$ from either $i$ or the successors of $i$ accessible by playing according to $\sigma$ at state $i$ (which coincides with the support of $P^{\sigma}_{i}$):
\[
T^{\sigma,\sigma'}_{i} \coloneq \left\{ j \in [n] \; | \; \exists h \in \{i\} \cup {\sf supp} \left(P^{\sigma}_{i}\right), \exists k \in \N, ((P^{\sigma'})^{k})_{hj} > 0 \right\} 
\]
where the support of a vector is the set of indices of the components that are nonzero. By the definition of $T_{i}^{\sigma,\sigma'}$ we get that the restriction $\tilde{P}^{\sigma'} \in \R^{|T_{i}^{\sigma,\sigma'}| \times |T_{i}^{\sigma,\sigma'}|}$ of $P^{\sigma'}$ to the states in $T_{i}^{\sigma,\sigma'}$ is stochastic.
We also consider the identity matrix $\tilde{I}$ and the Laplacian matrix $\tilde{L}^{\sigma'}$ restricted to $T_{i}^{\sigma,\sigma'}$, so $\tilde{I}, \tilde{L}^{\sigma'} \in \R^{|T_{i}^{\sigma,\sigma'}| \times |T_{i}^{\sigma,\sigma'}|}$. Using this notation, we define $\tilde{J}^{\sigma'} \coloneqq M \tilde{L}^{\sigma'}$ and $\tilde{D}^{\sigma'}_{i}(\rho) \coloneqq \det(\rho M \tilde{I} - \tilde{J}^{\sigma'})$. 
Then, we have the following property, which follows from~\cite[Theorem 2.1]{karlin1981second}.
\begin{lemma}\label{lem:positivity D tilde}

  Let $\lambda^{\sigma,\sigma'}_{i} \in \N$ be the multiplicity of $0$ as a root of $\tilde{D}^{\sigma'}_{i}(\rho)$. Then, $\lambda^{\sigma,\sigma'}_{i}$ is the number of recurrent classes of $P^{\sigma'}$ reachable from $\{i\} \cup {\sf supp} (P^{\sigma}_{i})$.
\end{lemma}
We are now ready to construct our family $\cP_d$:
\begin{equation}\label{eq:reduced family - d sens}
    \cP_d  = \left\{\rho \mapsto \frac{\tilde{D}^{\sigma'}_{i}(\rho)}{M \rho} C^{\sigma,\sigma'}_{i}(\rho) \; | \; i \in [n], \sigma \text{ and } \sigma' \text{ are adjacent at state } i,  \left( C^{\sigma,\sigma'}_{[-1,d+1]}\right)_{i} \neq 0 \right\}.
\end{equation}
The next lemma follows from the same analysis as in Section~\ref{sec:bounding height degree}.
\begin{lemma}\label{lem:height degree cPd}
Let $\Delta(\rho)=\tilde{D}^{\sigma'}_{i}(\rho) C^{\sigma,\sigma'}_{i}(\rho)/M\rho$, where the policies $\sigma$ and $\sigma'$ are adjacent at state $i\in [n]$. 
Then, $\Delta(\rho)$ is of the form $\sum_{k=0}^{n-1} b_k \rho^k$, where $b_k \in \Z$ and $|b_k| \leq 2 W M^{n-1} {n+1 \choose k+2}$.
\end{lemma}
\begin{proof}
This lemma follows by an application of \Cref{th:boundDeltaCoeff} to the MDP obtained restricting the state set to $T^{\sigma,\sigma'}_{i}$, noting that $C^{\sigma,\sigma'}_{i}(\rho)$ coincides with the $i$-th component of the comparison vector $\tilde{C}^{\sigma,\sigma'}(\rho)$ of this restricted MDP, along with the fact that the bounds in \Cref{th:boundDeltaCoeff} are non-decreasing in the number of states.  

\end{proof}
We are now ready to state the main result of this section, which shows that $\rhod$ is the root of a polynomial of the reduced family $\cP_d$, and shows some additional structure on the first nonzero coefficients of this polynomial.
\begin{theorem}\label{Th_d-sensitive_reduced}
Assume that the $\sens$-sensitive threshold $\rhod$ is such that $\rhod<+\infty$. Then, $\rhod$ is a positive root of a polynomial $\Delta(\rho)$ belonging to $\cP_d$ of the form $\sum_{k=j}^{n-1} b_k \rho^k$, with $b_j \neq 0$ and $j \leq d + \epsilon^{\sigma}_{i} + \lambda^{\sigma,\sigma'}_{i} - 1$,  where $\sigma$ and $\sigma'$ are policies adjacent at state $i \in [n]$.
\end{theorem}
\begin{proof}
A similar argument to the one used in the proof of \Cref{th:reduced family root bound alphabw} shows that there exist two stationary deterministic policies $\sigma$ and $\sigma\opt$ such that both of them are discount optimal at $\rhod$, $\sigma\opt$ is $d$-sensitive optimal, and $\sigma$ is not $d$-sensitive optimal. 
We now show that we can actually pick $\sigma$ adjacent to $\sigma\opt$ at some state $i$. 
With this aim, let us consider the sequence of policies $\sigma_0, \ldots , \sigma_n$, where $\sigma_{k}$ is obtained from $\sigma\opt$ by replacing, in the states in $\{1,\ldots ,k\}$, the action chosen by $\sigma\opt$ by the action chosen by $\sigma$, so that $\sigma_0 = \sigma\opt$ and $\sigma_n = \sigma$. Note that any two consecutive policies in $\sigma_0, \ldots , \sigma_n$ are adjacent, and any policy of this sequence is discount optimal at $\rhod$ (any policy that coincides with $\sigma$ or $\sigma\opt$ at all states is discount optimal). Besides, by \Cref{th:d sens lex ordering}, for any $k \in [n]$ we have $V^{\sigma_{k}}_{[-1,d]} - V^{\sigma\opt}_{[-1,d]} \succeq 0$  due to the $d$-sensitive optimality of $\sigma\opt$. Let $k'$ be the smallest $k \in [n]$ such that $V^{\sigma_{k}}_{[-1,d]} - V^{\sigma\opt}_{[-1,d]} \neq  0$. By minimality of $k'$, it follows that $V^{\sigma_{k'-1}}_{[-1,d]} = V^{\sigma\opt}_{[-1,d]}$, and so $\sigma_{k'-1}$ is $d$-sensitive optimal by \Cref{th:d sens lex ordering}, and that $V^{\sigma_{k'}}_{[-1,d]} - V^{\sigma_{k'-1}}_{[-1,d]} \succ 0$, and so $\sigma_{k'}$ is not $d$-sensitive optimal by the same theorem. Hence, the policies $\sigma_{k'}$ and $\sigma_{k'-1}$ are adjacent, both of them are discount optimal at $\rhod$, $\sigma_{k'-1}$ is $d$-sensitive optimal, and $\sigma_{k'}$ is not $d$-sensitive optimal. For the rest of the proof, we relabel $\sigma_{k'}$ and $\sigma_{k'-1}$ as $\sigma$ and $\sigma'$ respectively, and we let $i \in [n]$ be the state where they differ.

Let $\Delta(\rho) = \tilde{D}^{\sigma'}_{i}(\rho)C^{\sigma,\sigma'}_{i}(\rho)/M\rho$. Note that $\Delta(\rhod)=0$ since $C^{\sigma,\sigma'}_{i}(\rhod)=0$ (both $\sigma$ and $\sigma'$ are discount optimal at $\rhod$). The fact that $\Delta(\rho)$ is of the form $\sum_{k=0}^{n-1} b_k \rho^k$ follows from Lemma~\ref{lem:height degree cPd}. 
Now let $d'$ be the smallest  $d\in \{-1,0,1,\ldots\}$ such that $(C^{\sigma,\sigma'}_{[d]})_i \neq 0$. By definition of $d'$, we have $C^{\sigma,\sigma'}_{i}(\rho) = \rho^{d'}g(\rho)$ for some Laurent series $g(\rho)$ such that $g(0) \neq 0$. Moreover, by definition of $\lambda^{\sigma,\sigma'}_{i}$ we have $\tilde{D}^{\sigma'}_{i}(\rho) = \rho^{\lambda^{\sigma,\sigma'}_{i}}p(\rho)$, where $p(\rho)$ is a polynomial such that $p(0) \neq 0$. 
It follows that $\Delta(\rho) = \rho^{d'+\lambda^{\sigma,\sigma'}_{i}-1}p(\rho)g(\rho)$ with $p(0)g(0) \neq 0$, which  implies $b_0=\ldots = b_{d'+\lambda^{\sigma,\sigma'}_{i}-2}=0$ and $b_{d'+\lambda^{\sigma,\sigma'}_{i}-1} \neq 0$, and so $j = d'+\lambda^{\sigma,\sigma'}_{i}-1$. We now show $d' \leq k + \epsilon^{\sigma}$ for some $k \in \N, k \leq d$.
Since $V^{\sigma}_{[-1,d]} - V^{\sigma'}_{[-1,d]} \succ 0$, 
for some state $h$ the first nonzero Laurent coefficient of $(v^{\sigma}_{\rho} - v^{\sigma'}_{\rho})_h$ has an index $k\leq d$ and is positive. Besides, we have $(v^{\sigma}_{\rho} - v^{\sigma'}_{\rho})_h = R^{\sigma}_{hi}(\rho)C^{\sigma,\sigma'}_{i}(\rho)$ because $\sigma$ and $\sigma'$ are adjacent at state $i$. Thus, $R^{\sigma}_{hi}$ is not identically zero because $(v^{\sigma}_{\rho} - v^{\sigma'}_{\rho})_h$ is not identically zero. 
By \Cref{lem:resolvent properties} we know that the first nonzero index of the Laurent expansion of $R^{\sigma}_{hi}$ is $-\epsilon^{\sigma}_{i}$, with a positive leading coefficient. It follows that $d'=k+\epsilon^{\sigma}_{i}$, and so $d'\leq d +1$ and $j = d'+\lambda^{\sigma,\sigma'}_{i}-1\leq d + \epsilon^{\sigma}_{i} + \lambda^{\sigma,\sigma'}_{i} - 1$ because $k\leq d$ and $\epsilon^{\sigma}_{i} \in \{0,1\}$. Finally, since $d'\leq d +1$, by definition of $d'$ we have $(C^{\sigma,\sigma'}_{[-1,d+1]})_i\neq 0$, which implies that $\Delta(\rho)$ belongs to $\cP_d$. 

\end{proof}
To obtain lower bounds on $\rhod$ it is sufficient to lower bound the positive roots of the polynomial $\Delta(\rho)$ that appears in Theorem~\ref{Th_d-sensitive_reduced}.  
The fact that some of the first coefficients of $\Delta(\rho)$ are zeros can be leveraged to prove stronger separation results, still relying on Cauchy's bound.
We get the following separation result for the roots of elements in $\cP_d$, based on a variant of Cauchy's bound.
\begin{lemma}\label{lem:ZerosDelta reduced}
Assume that $\Delta(\rho) \in \cP_d$ is of the form $\sum_{k=j}^{n-1} b_k \rho^k$, with $b_j \neq 0$ for some $j \in \{0,...,n-1\}$ (so $\Delta(\rho)$ is not the zero polynomial). If $j = n-1$, then $\Delta(\rho)$ has no positive roots. Otherwise, $\Delta(\rho)$ has no roots in the interval $\left(0, \frac{1}{4 W M^{n-1}{n+1\choose j + 3}}\right)$.
\end{lemma}
\begin{proof} 
If $j=n-1$, then the unique root of $\Delta(\rho)$ is $0$. Otherwise, we can factor $\Delta(\rho)$ as $\Delta(\rho) = \rho^{j} q(\rho)$ with $q(\rho) = b_{j} + \sum_{r=1}^{n-j-1} b_{j+r}\rho^r$. By \Cref{lem:height degree cPd}, we know that $|b_{k}| \leq B_{k}$, where $B_{k} = 2 WM^{n-1} {n+1 \choose k+2}$ for $k \in \{j,...,n-1\}$. Thus, for $r \in \{1,...,n-j-1\}$ we have $B_{j+r} \leq B^{r}_{j+1}$, since $2 WM^{n-1} \geq 1$ and, from the definition of binomial coefficients, we have ${n+1 \choose j+r + 2} \leq {n+1 \choose j+3}^{r}$ for $r > 1$. 
Then, if $z \in (0,1/2B_{j+1})$ is a root of $\Delta(\rho)$, we have 
\[
1 \leq |b_j| = | \sum_{r =1}^{n-j-1} b_{j+r}z^r | \leq \sum_{r =1}^{n-j-1} |b_{j+r}|z^r \leq \sum_{r =1}^{n-j-1} |B_{j+1}z|^{r} < \sum_{r =1}^{n-j-1} \frac{1}{2^r} < 1 \; ,
\]
which is a contradiction. Hence, any positive root $z$ of $\Delta(\rho)$ satisfies $z \geq \frac{1}{4 W M^{n-1} {n+1 \choose j+3}}$. 
\end{proof}
We can now prove the main theorem of this section. Given a MDP $\cM$, we define $\Lambda(\cM) \coloneqq \max \{ \epsilon^{\sigma}_{i} +  \lambda^{\sigma,\sigma'}_{i}  \; | i \in [n], \sigma \text{ and } \sigma' \text{ are adjacent at state } i\}$, and given a game $\Gamma$, we define $\Lambda(\Gamma) \coloneqq \max \{\Lambda(\cM) \; | \; \cM \in \cM(\Gamma)\}$, where $\cM(\Gamma)$ is the set of MDPs induced by $\Gamma$ defined in \Cref{sec:games to mdps}.
\begin{theorem}\label{Th:d-sensitiveBound}
The $d$-sensitive threshold $\rhod$ of a two-player perfect-information stochastic game satisfying Assumption \ref{assumption} verifies
  \[
\frac{1}{4 W M^{n-1} {n+1 \choose {\min \{ d+\Lambda(\Gamma)+2, \lfloor  (n+1)/2\rfloor \} }}}  \leq  \rho_{\text{d}}.
  \]
\end{theorem} 
\begin{proof}
By \Cref{prop:reduction games to mdps}, it is enough to show that the given bound holds for the $\sens$-sensitive threshold $\rhod$ of any MDP $\cM$ in $\cM(\Gamma)$, so let $\cM \in \cM(\Gamma)$. If $\rhod < +\infty$ (otherwise, there is nothing to prove), by \Cref{Th_d-sensitive_reduced} we know that $\rhod$ is a positive root of a polynomial $\Delta(\rho)$ in $\cP_d$. Let $\sigma$ and $\sigma'$ be adjacent policies differing at some state $i$ such that $\Delta(\rho)= \tilde{D}^{\sigma'}_{i}(\rho)C^{\sigma,\sigma'}_{i}(\rho)/M\rho$. Then, from \Cref{lem:ZerosDelta reduced} and \Cref{Th_d-sensitive_reduced}, we conclude that
$
\frac{1}{4 W M^{n-1} {n+1 \choose {\min \{ d+\epsilon^{\sigma}_{i} + \lambda^{\sigma,\sigma'}_{i}+2, \lfloor  (n+1)/2\rfloor \} }}}  \leq  \rho_{\text{d}}.
$
  The result now follows from the fact that $\epsilon^{\sigma}_{i} +\lambda^{\sigma,\sigma'}_{i}\leq \Lambda(\Gamma)$. 

\end{proof}
\subsection{The case of deterministic instances.}\label{sec:deter d sens}
Our results in the previous section can be improved for deterministic instances.
\paragraph{Sharper bounds on $\rhod$.}
For deterministic instances, we have a better control on the $\epsilon^{\sigma}_{i}$ and $\lambda^{\sigma,\sigma\opt}_{i}$, as we show next.
\begin{lemma}\label{lem:bound lambda}
Let $\cM$ be an MDP satisfying Assumption~\ref{assumption} with $M=1$.
If $\sigma$ and $\sigma'$ are policies adjacent at state $i$, then $\epsilon^{\sigma}_{i} + \lambda^{\sigma,\sigma'}_{i} \leq 2$.
\end{lemma}
\begin{proof}
By \Cref{lem:positivity D tilde}, $\lambda^{\sigma,\sigma'}_{i}$ is the number of recurrent classes under $P^{\sigma'}$ starting from $\{i\} \cup {\sf supp}(P^{\sigma}_{i})$. Since the MDP is deterministic, we have ${\sf supp}(P^{\sigma}_{i}) = \{\sigma(i)\}$, identifying $\sigma(i)$ with the successor of $i$ chosen by $\sigma$. Recall that $\epsilon^{\sigma}_{i}$ is $1$ if state $i$ is recurrent under $P^{\sigma}$ and $0$ otherwise.  
We consider the two cases $\epsilon^{\sigma}_{i}=1$ and $\epsilon^{\sigma}_{i}=0$.

If $\epsilon^{\sigma}_{i}=1$ then $i$ is recurrent under $P^ \sigma$, i.e., there is a path that starts in $i$, ends in $i$, and does not visit $i$ in between the start and end states. Note that the second state in this path is $\sigma(i)$. 
Since $\sigma$ and $\sigma'$ coincide everywhere except at state $i$, it follows that either $\sigma(i) = i$ or there is a path from $\sigma(i)$ to $i$ by playing $\sigma'$. We conclude that under $P^{\sigma'}$, starting from $i$ or from $\sigma(i)$, we reach the same recurrent class. Thus $\lambda^{\sigma,\sigma'}_{i}=1$, and so $\epsilon^{\sigma}_{i} + \lambda^{\sigma,\sigma'}_{i} = 2$.

If $\epsilon^{\sigma}_{i}=0$, we still have that $\lambda^{\sigma,\sigma'}_{i}$ is the number of recurrent classes under $P^{\sigma'}$ starting from a set containing at most two states. Since the number of recurrent classes under $P^{\sigma'}$ starting from a single state is $1$ (the MDP is deterministic), we get that $\lambda^{\sigma,\sigma'}_{i} \leq 2$, and so $\epsilon^{\sigma}_{i} + \lambda^{\sigma,\sigma'}_{i} \leq 2$.  
\end{proof}
We get the following corollary by applying Lemma~\ref{lem:bound lambda} and Theorem~\ref{Th:d-sensitiveBound}.
\begin{corollary}\label{cor:d-sensitiveBound - deter}
The $d$-sensitive threshold $\rhod$ of a two-player perfect-information deterministic stochastic game satisfying Assumption \ref{assumption} with $M=1$ verifies
  \[
\frac{1}{4 W  {n+1 \choose {\min \{ d+4, \lfloor  (n+1)/2\rfloor \} }}}  \leq  \rho_{\text{d}}.
  \]
\end{corollary}
Note that for $d=-1$ this recovers the bound $\rho_{-1}^{-1} = O(W n^3)$ obtained in~\cite{zwick} using ${n+1 \choose 3} \leq n^3$.
\paragraph{The multiplicity approach for $\rhobw$.}
The following approach was developed in~\cite{gaubert2026thresholds} for deterministic MDPs, based on the complete comparison family induced by the difference of all value functions. We now show how to adapt this method to our reduced family.

Recall that in general, Blackwell optimal policies correspond to policies that are $d$-sensitive for $d = n-2$. An interesting property of deterministic instances is that Blackwell optimality may be achieved by $d$-sensitive optimal policies with $d$ potentially much smaller than $n-2$, as we show next. 
The following definition extends the notion of multiplicity at $0$ (usually defined for polynomials) to rational functions.
\begin{definition}
Let $F:\R_{+}^{*} \to \R$ be a meromorphic function.
We say that $F$ has order $\ell \in \Z$ at $0$ if $F(\rho)=\rho^{\ell}f(\rho)$ for $\rho >0$, where $f$ is an  analytic function at $0$ such that $f(0) \neq 0$.
\end{definition}
We note that the order of $F$ at $0$ is just the exponent of the first non-zero term of the Laurent series expansion of $F$ at $0$. We have the next proposition.
\begin{proposition}\label{prop:bar d implies blackwell}
Assume that $\bar{d} \in \N$ is an upper bound on the order at $0$ of every nonzero $C^{\sigma,\sigma'}_{i}(\rho)$, where $\sigma$ and $\sigma'$ are any policies adjacent at some state $i \in [n]$. Then, any $\bar{d}$-sensitive optimal policy is also Blackwell optimal.
\end{proposition}
\begin{proof}
Let $\sigma\opt$ be $\bar{d}$-sensitive optimal. To prove that $\sigma\opt$ is Blackwell optimal, by Proposition~\ref{PropVeinott} it is enough to show that for any policy $\sigma$ adjacent to $\sigma\opt$ at some state $i$, we have $C^{\sigma,\sigma\opt}_{i}(\rho)\geq 0$ for $\rho>0$ sufficiently small. So let $\sigma$ be any policy adjacent to $\sigma\opt$ at some state $i$. If $C^{\sigma,\sigma\opt}_{i}(\rho)=0$ for $\rho>0$, then there is nothing to prove. 
Otherwise, from Lemma \ref{lem:resolvent properties} we have $(v^{\sigma}_{\rho} - v^{\sigma\opt}_{\rho})_i = R^{\sigma}_{ii}(\rho)C^{\sigma,\sigma\opt}_{i}(\rho)$, with $R^{\sigma}_{ii}(\rho)>0$ for $\rho>0$ and $R^{\sigma}_{ii}$ has order $0$ or $-1$ at $\rho=0$ (depending if $i$ is transient or recurrent under $P^{\sigma}$). In both cases, the leading coefficient in the Laurent series of $R^{\sigma}_{ii}$ is positive. Therefore, the first nonzero coefficient in the Laurent series of $(v^{\sigma}_{\rho} - v^{\sigma\opt}_{\rho})_i$ has the same sign as the first nonzero coefficient of the Laurent series of $C^{\sigma,\sigma\opt}_{i}$, and its order is at most the order of $C^{\sigma,\sigma\opt}_{i}$, which is at most $\bar{d}$ by assumption. Since $\sigma\opt$ is $\bar{d}$-sensitive optimal we have $\liminf_{\rho \to 0^+} \rho^{-\bar{d}}(v^{\sigma}_{\rho} - v^{\sigma\opt}_{\rho})_i \geq 0$ and since the first nonzero coefficient in the Laurent series of $(v^{\sigma}_{\rho} - v^{\sigma\opt}_{\rho})_i$ is of order at most $\bar{d}$, this coefficient must be positive, and the same is true for the first nonzero coefficient in the Laurent series of $C^{\sigma,\sigma\opt}_{i}$. Hence $C^{\sigma,\sigma\opt}_{i}(\rho)>0$ for $\rho>0$ small enough, which concludes the proof. 
\end{proof}
\begin{remark}
The previous proposition does not assume the MDP is deterministic. But we will need this assumption to obtain stronger bounds on the coefficients of the comparison vectors.
\end{remark}
We are now left with bounding the order at $0$ of the comparison vectors. For deterministic instances, by definition of the comparison vectors, we have
\begin{equation}\label{eq:comp vector - deterministic}
C^{\sigma,\sigma\opt}_{i}(\rho) = r^{\sigma(i)}_{i} + (v^{\sigma\opt}_{\rho})_{\sigma(i)} - (1+\rho)(v^{\sigma\opt}_{\rho})_i
\end{equation}
where, with a slight abuse of notation, we denote by $\sigma(i)$ the successor of $i$ when playing according to the stationary deterministic policy $\sigma$ at state $i$.
The value functions $(v^{\sigma}_{\rho})_i$ have a very particular structure in deterministic instances: every trajectory consists of an initial phase visiting transient states, then a cycle. We use this structure to arrive at the next proposition.
\begin{proposition}\label{prop:comp vector height and degree}
Assume that the MDP instance satisfies Assumption~\ref{assumption} with $M=1$. Let $\sigma$ and $\sigma'$ be two stationary deterministic policies and $i \in [n]$. Then, we can write 
\[C^{\sigma,\sigma'}_{i}(\rho) = F(\rho)/G(\rho),  \quad \forall \; \rho>0,\]
where $G$ has order at most $2$ at $0$, $G(\rho)>0$ for $\rho>0$, and $F(\rho)=\sum_{k=0}^{n} a_k (1+\rho)^{-k}$ with $a_k \in \Z$ and $|a_k| \leq  8W$ for any $k \in \{0,...,n\}$.
\end{proposition}
We show that our bounds on the coefficients of $F$ are tight in Appendix \ref{app:deter d sens}.
From the previous proposition, a bound on the orders of the comparison vectors can be obtained by bounding the order of the functions $F$ appearing in Proposition~\ref{prop:comp vector height and degree}. We can bound this order using the following lemma, which is an easy extension of~\cite{BEK1999}.
\begin{lemma}\label{lem:BEK99 at 0}
There exists a constant $A>0$ such that for any $F(\rho)=\sum_{k=0}^{N} a_k (1+\rho)^{-k}$ that is non-identically zero with integer coefficients bounded in absolute value by $H$, we have that the order of $F$ at $0$ is bounded by $A \sqrt{N(1+\log(H))}$.
\end{lemma}
Combining Proposition~\ref{prop:comp vector height and degree} and Lemma~\ref{lem:BEK99 at 0}, the order at $0$ of any comparison vector is bounded by $\bar{d}=O\left(\sqrt{n (1+ \log(W))}\right)$. Plugging this into the bound from Corollary~\ref{cor:d-sensitiveBound - deter} yields the next corollary.
\begin{corollary}\label{cor:d-sensitiveBound - deter - plugged}
The Blackwell threshold $\rhobw$ of a two-player perfect-information deterministic stochastic game satisfying Assumption \ref{assumption} with $M=1$ verifies 
  \[
\frac{1}{4 W {n+1 \choose {\min \{ \bar{d}+4, \lfloor  (n+1)/2\rfloor \} }}}  \leq  \rhobw \; ,
  \]
where $\bar{d}=O\left(\sqrt{n (1+ \log(W))}\right)$.
\end{corollary}

\begin{remark}
An application of this method for non-deterministic instances yields bounds on $\bar{d}$ that are meaningless. Indeed in the case of stochastic instances, the term $\log(H)$ from Lemma~\ref{lem:BEK99 at 0} is of order $\Theta(n)$ so that the bound on the order $\bar{d}$ is of order $\Theta(n)$, and we already know that $\bar{d} \leq n$. 
\end{remark}
\paragraph{Sharper bounds based on Mahler.}
The next theorem can be obtained by combining \Cref{prop:comp vector height and degree} and \Cref{th:mahler sep - generic}.
\begin{theorem}\label{th:dubickas - deter}
For each $\epsilon > 0$ there exists a constant $c_\epsilon>0$ such that the Blackwell threshold $\rhobw$ of a two-player perfect-information stochastic game satisfying Assumption~\ref{assumption} with $M=1$ verifies, for $\theta \coloneqq 1 + \log(2)$ and $\displaystyle B \coloneqq 8 W \sqrt{n+1}$,
\[\min \left\{ \exp\left(-\left(\frac{\pi}{4}+\epsilon \right) \sqrt{n\log(n)\log(B)}\right), \frac{1}{c_{\epsilon} B^{\theta}}\right\} \leq \rhobw.\]
\end{theorem}
\begin{proof}
From \Cref{prop:comp vector height and degree}, to bound $\alphabw$ we can use the map $\tilde{F}:\alpha \mapsto F((1-\alpha)/\alpha) = \sum_{k=0}^{n} a_{k}\alpha^k$ with $|a_k| \leq 8W$. Now we can use \Cref{th:mahler sep - generic} to separate a root $z$ of $\tilde{F}$ of degree $d$ from $1$, noting that $d\leq n$ and $\sfM(z) \leq \sfM(\tilde{F}) \leq 8W\sqrt{n+1}$. 
\end{proof}
\section{Discussion.}\label{sec:discussion}
\paragraph{More comparison with prior work on $\rhobw$.} We first provide more context on prior work studying bounds on $\alphabw$.
While the existence of $\alphabw$ is shown for the first time in~\cite{blackwell1962discrete} (for MDPs) and~\cite{liggettlippman} (for SGs), to the best of our knowledge~\cite{andersson2009complexity} is the first to provide a separation result for $\alphabw$ and $1$ in SGs using root separation techniques (\cite{andersson2009complexity} uses Cauchy's bound). We note that the authors of~\cite{andersson2009complexity} do not try to get the tightest bound since they are only interested in obtaining a polynomial reduction from mean-payoff optimality to discount optimality, and their bound suffices for this purpose (it gives that $1-\alphabw$ has polynomial bit-size under Assumption~\ref{assumption}).
\cite{grand2026beyond} uses a construction related to~\cite{andersson2009complexity} but for MDPs, and using a less specialized separation result (Theorem~3 of~\cite{rump1979polynomial}, which separates any two algebraic numbers, while Cauchy is tailored to separation from $0$).
The two works~\cite{asadi2024deterministic, mukherjee2025howard} are very close and focus on deterministic instances, both  relying on~\cite{borwein1997polynomials} to bound the multiplicity of polynomials at $1$, and use this to bound $\alphabw$ directly (without studying $\alphad$ or showing a bound on $\bar{d}$ as defined in \Cref{prop:bar d implies blackwell}). \cite{asadi2024deterministic, mukherjee2025howard} arrive at the same bounds for $\alphabw$, and \cite{asadi2024deterministic} additionally provides lower bounds for deterministic instances. For exactness, we note that \cite{asadi2024deterministic, mukherjee2025howard} focus on studying the complexity of Strategy Iteration and only obtain bounds on $\alphabw$ as a byproduct of their analysis; we believe that their lines of arguments can be adapted to reduce the degree of the polynomials they study by $1$.
\cite{gaubert2026thresholds} focuses on SGs and introduce the separation method based on Mahler measures for $\alphabw$ in their Theorems~3.8 and~4.5, but these results need to be fixed to account for the small height assumption necessary to use the results from \cite{Dubickas1995} (as we highlighted in \Cref{sec:root separation}); \Cref{th:mahler sep - generic} provides such a fix and \Cref{th-dubickas} applies this corrected separation bound for bounding the Blackwell threshold. We also note that~\cite{loff2024smoothed} obtains a bound on $\alphabw$ for deterministic instances but only control it in probability (e.g. Theorems~5.7 and~6.1 in~\cite{loff2024smoothed}) whereas we are interested in universal bounds on the Blackwell threshold. Finally, we note that there is a stream of work using root separation results for the discounted values and mean-payoff values as algebraic numbers, see e.g.~\cite{oliu2021new,akian2019operator,hansen2011exact}. In contrast to these works, we use root separation results for the discount factor and interest rate.

To simplify the comparisons among the many bounds obtained in the past literature, we find it easier to directly compare the comparison families used in each work. This comparison is presented in Table~\ref{tab:comparison}.  
To be able to fairly compare our results with prior works, it is crucial to note that many prior works~\cite{grand2026beyond,asadi2024deterministic,mukherjee2025howard,gaubert2026thresholds} have worked with polynomials in the variable $\alpha$, so that they need to separate their positive roots from $1$, whereas we work with polynomials in the variable $\rho$ and we need to separate their positive roots from $0$. 
This is why in Table~\ref{tab:comparison} we separate the comparison with prior works depending if they work with separation from $1$ (as in~\cite{grand2026beyond,asadi2024deterministic,mukherjee2025howard,gaubert2026thresholds}) or separate from $0$ (as in~\cite{andersson2009complexity} and our work).
We also isolate the results for deterministic instances since in this case more structure can be exploited to bound the coefficients of the polynomials in the comparison family (we give our results in this case in Section~\ref{sec:deter d sens}).
We see that our reduced family has the smallest degree among comparison families studied in prior work, and our bounds on the coefficients are in fact also much better (keeping in mind that prior works separating positive roots of polynomials from $1$ would need an extra binomial coefficient to be translated as separation from $0$, which would make their bounds even worse, see \Cref{lem:cauchy from 1}). 
Overall, we note that we obtain the best bounds on the degree and height of our comparison family (as in \Cref{th:reduced family bound}), and we are the first to show that these bounds are tight (\Cref{th:reduced family bound - tightness}). This is made possible by (a) working with comparison vectors to lower the degree of the polynomials studied, and (b) working with the variable $\rho$ instead of $\alpha$, so that the Laplacian $L^{\sigma}$ appears naturally in our expression and we can rely on the Matrix Forest theorem. As we explained at the end of \Cref{sec:upper bound rho}, our upper bound on $\rhobw$ from \Cref{th:lower bound deterministic} improves in the asymptotics in $n$ compared to the bounds from~\cite{asadi2024deterministic}. 

\begin{table}[htp]
\caption{Comparison families obtained in our paper and in previous works. The columns ``Degree" and ``Height" give upper bound on the degree and height of the comparison family computed in the relevant works.}
    \label{tab:comparison}
    \centering
    \begin{tabular}{c|cccccc}
\toprule
    Paper  & Comp. family & Variable & Separate from & Degree &  Height & Remarks \\
        \midrule
    \cite{asadi2024deterministic} & complete & $\alpha$ & $1$ & $2n$  & $12 W$ & Deterministic \\
\cite{mukherjee2025howard} & complete & $\alpha$ & $1$ & $2n+1$  & $12 W$ & Deterministic
    \\  
     \cite{gaubert2026thresholds} & complete & $\alpha$ & $1$ & $2n-1$  & $12 W$ & Deterministic \\
     This paper & reduced  & $\alpha$  &  $1$  & $n$ & $8W$ & Deterministic \\
        \hline
    \cite{grand2023reducing} & complete & $\alpha$ & $1$  & $2n-1$ & $2W n 4^{n} M^{2n}$  &  MDPs only \\

    \cite{gaubert2026thresholds} & complete & $\alpha$  & $1$  & $2n-1$& $2W n M^{2n-1} {2n-1 \choose n }  $  & SGs \\
        \hline
        \hline

    \cite{andersson2009complexity} & complete & $1-\alpha$ & $0$ & $2n-1$  & $2 W (n!)^2 4^{n} M^{2n^2}$ &  SGs \\
      This paper & reduced & $\rho$ & $0$  & $n-1$ & $2 W M^{n-1} {n+1 \choose \lfloor (n+1)/2 \rfloor}$ &  Tight (Th.~\ref{th:reduced family bound - tightness})  \\
        \bottomrule
    \end{tabular}
\end{table}

We also find it worthy to emphasize the simplicity of our bounds. Our Theorem~\ref{Th:boundBlackwellCauchy} gives $\rhobw^{-1} \leq 4 W M^{n-1} {n+1 \choose \lfloor (n+1)/2 \rfloor}$ while Corollary 4.2 in~\cite{gaubert2026thresholds} gives 
\[(1-\alphabw)^{-1} \le   \frac{nW (2M)^{2n-1} {2n-1 \choose \lfloor \frac{2}{3} n \rfloor}}{2^{\lfloor \frac{2}{3} n \rfloor-2}}\]
and Theorem 5.2 in~\cite{mukherjee2025howard} shows that for deterministic instances and $\bar{d} = O\left( \sqrt{(2n+1) \log(12 W)} \right)$, we have
\[(1- \alphabw)^{-1} \le W \cdot O\left(  \frac{ e^{\bar{d}} (2n+1)^{\bar{d}+2}} {\bar{d}^{\bar{d}}} \right).\]
\paragraph{More comparison with prior work on $\alphad$.}  \cite{gaubert2026thresholds} also provides bounds for $\alphad$ for multichain deterministic instances, but in the case of stochastic instances \cite{gaubert2026thresholds} is only able to bound $\alphad$ under the unichain assumption, whereas we are able to do so for multichain stochastic instances in \Cref{Th:d-sensitiveBound}. Our proof of the bound on $\bar{d}$ such that $\bar{d}$-sensitive optimal policies are Blackwell optimal (\Cref{prop:bar d implies blackwell}) follows the same lines as the one of Theorem~3.5 in~\cite{gaubert2026thresholds}, but our bound on $\bar{d}$ is much better given we greatly improve the degree and height in the comparison family.
\paragraph{Comparisons of our bounds.}
The bounds obtained on the Blackwell thresholds based on Cauchy's bound, Mahler measure, and multiplicity approach are difficult to compare and which bound is best depends on the parameter regimes for $n$, $W$ and $M$. As argued at the end of \Cref{sec:separation blackwell}, for stochastic instances the bound based on Mahler from \Cref{th-dubickas} is always worse than the bound based on Cauchy from \Cref{Th:boundBlackwellCauchy}, while the multiplicity approach does not work for stochastic instances. For deterministic instances the picture is more mixed, so we provide a comparison in \Cref{tab:deterministic-bound-comparison-logW}; we see that bounds based on Mahler, Cauchy or multiplicity can all be optimal for some regimes of $\log(W)$, thus no separation method uniformly dominates the others. Note that the multiplicity bound is always (weakly) stronger than Cauchy's bound (because of the minimum in the bound of \Cref{cor:d-sensitiveBound - deter - plugged}) so in \Cref{tab:deterministic-bound-comparison-logW} we distinguish when Cauchy matches the multiplicity bound (in which case we indicate that ``Cauchy is the best method" in the last column) or not (in which case we indicate that ``Multiplicity is the best method" in the last column).

\begin{table}[t]
\caption{Best deterministic upper bounds on
$\log(\rhobw^{-1})$ in the different regimes of $\log(W)$. For the sake of readability we omit the multiplicative constants in the entries above. We write $c'=\frac{2(\pi/4+\varepsilon)}{A}$ with $A>0$ the constant appearing in Lemma \ref{lem:BEK99 at 0} and $c''=c'/2$. For two real sequences $(a_n)_{n \in \N},(b_n)_{n \in \N}$ we write $a_n \ll b_n$ for $a_n = o(b_n)$.}
\centering
\small
\renewcommand{\arraystretch}{1.55}
\setlength{\tabcolsep}{5pt}
\begin{tabularx}{\textwidth}{
    >{\raggedright\arraybackslash}p{5.5cm}
    >{\centering\arraybackslash}X
    >{\centering\arraybackslash}p{4.0cm}
}
\toprule
Regime for $\log(W)$
&
Best $O(.)$ upper bound on $\log(\rhobw^{-1})$
&
Best methods
\\
\midrule

$\log W=O(1)$
&
$\sqrt{n}\log(n)
$
&
Mahler (Th. \ref{th-dubickas}) and Multiplicity (Cor. \ref{cor:d-sensitiveBound - deter - plugged})
\\[1ex]

\hline

$\displaystyle 1 \ll \log W \ll n\exp\left(-c'\sqrt{\log(n)}\right)$
&
$\displaystyle\sqrt{n\log(n)(\log W+\log(n))}$
&
Mahler (Th. \ref{th-dubickas})
\\[2ex]

$\displaystyle \log W = \Theta\left( n\exp\left(-c'\sqrt{\log(n)}\right)\right)$
&
$\displaystyle n \sqrt{\log(n)}\exp(-c'' \sqrt{\log(n)})$
&
Mahler (Th. \ref{th-dubickas}) and Multiplicity (Cor. \ref{cor:d-sensitiveBound - deter - plugged})
\\[2ex]

$\displaystyle
n\exp\left(-c'\sqrt{\log(n)}\right)\ll \log(W) \ll n
$
&
$\displaystyle
\sqrt{n\log(W)}\,
\log\frac{n}{\log(W)}
$
&
Multiplicity
(Cor. \ref{cor:d-sensitiveBound - deter - plugged})
\\[2ex]

$ \log(W) = \Theta(n)
$
&
$n
$
&
Multiplicity (Cor. \ref{cor:d-sensitiveBound - deter - plugged}) and Cauchy (Th. \ref{Th:boundBlackwellCauchy})
\\[2ex]

$\displaystyle n \ll \log(W) $
&
$\displaystyle \log(W) + n$
&
Cauchy (Th. \ref{Th:boundBlackwellCauchy})
\\
\bottomrule
\end{tabularx}
\label{tab:deterministic-bound-comparison-logW}
\end{table}
\paragraph{Future work.}
Closing the gap between the upper and lower bounds on $\rhobw$ for deterministic instances (the upper bound is $\rhobw^{-1} = O(\exp(c  \sqrt{n} \log(n))$ and the lower bound is $\rhobw^{-1} = \Omega(\exp(c (n\log(n))^{1/3})$ (for fixed $W \in \N$) from Theorem~\ref{th:lower bound deterministic}) is an important problem. One can also study more restricted classes of games to obtain stronger results, e.g., irreducible games as in~\cite{boros2013discounted}. For more general classes of games such as concurrent (non-perfect information) stochastic games it is already known that we may have $1/\rho_{-1}$ to be double-exponential in the number of states of the game~\cite{miltersen2011discounted} so it may be difficult to obtain any useful bound on $\rhobw$.
Since our bounds for the polynomials in the reduced family are tight both for stochastic and deterministic instances (Theorem~\ref{th:reduced family bound - tightness} and after Proposition~\ref{prop:comp vector height and degree}), it is worth noting that to obtain improved bounds one must either exploit more structure for the root separation results, or use completely new ideas not based on comparison families (in contrast to all the past literature). 
\section{Conclusion.}
In this paper, we develop a unified approach for bounding the Blackwell and $d$-sensitive thresholds of perfect-information stochastic games. Our main contributions are the construction of reduced comparison families of polynomials based on adjacent policies, and a novel matrix-forest analysis that yields sharper (and essentially tight) estimates of the degree and height of the polynomials. These structural results are then applied to yield improved bounds on the Blackwell thresholds and the first general bounds for the $d$-sensitive thresholds in the multichain stochastic case, refined in the case of deterministic instances. The main open questions left by our work are to close the gap between the upper and lower bounds for the Blackwell thresholds (especially for deterministic instances) as well as to identify subclasses of games where the Blackwell and $d$-sensitive thresholds are ``only'' polynomially close to $1$.
\bibliography{tropical}
\bibliographystyle{alpha}
\appendix
\section{Proof of \Cref{prop:reduction games to mdps}.}\label{app:proof mdp reduction}
\begin{proof}[Proof of \Cref{prop:reduction games to mdps}.]
(a) This claim follows from the definition of the MDPs induced by a stochastic game.

(b) Let $(\sigma\opt,\tau\opt)$ be a pair of discount optimal policies in $\Gamma$ for an interest rate $\rho'>0$ satisfying $\rho'<\bar{\rho}$. To show that $\bar{\rho} \leq \rhobw$ we need to prove that $(\sigma\opt,\tau\opt)$ is Blackwell optimal in $\Gamma$. As $(\sigma\opt,\tau\opt)$ is discount optimal in $\Gamma$ for $\rho'$, by~\eqref{eq:DefDiscOpt} we have that $\sigma\opt$ (resp. $\tau\opt$) is discount optimal in $\Gamma_{\tau\opt}^{\Min}$ (resp. $\Gamma_{\sigma\opt}^{\Max}$) for $\rho'$. Since $\rho'<\bar{\rho}$ and the Blackwell threshold of any MDP in $\cM(\Gamma)$ is lower bounded by $\bar{\rho}$, it follows that $\sigma\opt$ (resp. $\tau\opt$) is discount optimal in $\Gamma_{\tau\opt}^{\Min}$ (resp. $\Gamma_{\sigma\opt}^{\Max}$) for $\rho>0$ small enough. Thus, $(\sigma\opt,\tau\opt)$ is discount optimal in $\Gamma$ for $\rho>0$ small enough, and so it is Blackwell optimal in $\Gamma$.

(c) The argument for proving this claim is similar to the one used to prove (b), and so we omit it for conciseness. 
\end{proof}

\section{Comparison with lower bounds in~\cite{mukherjee2025efficient}.}\label{app:comparison mukherjee}
We compare our upper bound for stochastic MDP instances from Section~\ref{sec:lower bound sto} with the results from~\cite{mukherjee2025efficient}.
\begin{proposition}[Theorem 2, \cite{mukherjee2025efficient}]
Let $\epsilon \in (0,1/2)$.
There exists an MDP with rewards in $\{0,1\}$, transitions in $\displaystyle \{0,1,1/2,1/2-\epsilon,\epsilon\}$ and $n+7$ states such that $\alphabw \geq 1 - \frac{1}{\epsilon 2^{n/3}}.$
\end{proposition}
For $\epsilon=1/q$ with $q \geq 2, q \in \N$ we have that $M=2q=2\epsilon^{-1}$ in an MDP instance with transitions in $\{0,1,1/2,1/2-\epsilon,\epsilon\}$, and the authors in~\cite{mukherjee2025efficient} must take $\epsilon \in (0,1/2)$ to maintain non-negative transition probabilities. The term $\frac{1}{\epsilon 2^{n/3}}$ then becomes $\frac{M}{2}\frac{1}{2^{n/3}}$, as we claim in Section~\ref{sec:lower bound sto}.
\section{Proof of Theorems~\ref{th:borwein pinner - coefficient H} and~\ref{th:lower bound deterministic}.}\label{app:lower bound deterministic}
\Cref{th:lower bound deterministic} follows directly from applying \Cref{lem:reduction poly to alphabw} to the polynomials constructed in \Cref{th:borwein pinner - coefficient H}. The rest of this section is devoted to proving \Cref{th:borwein pinner - coefficient H}.
The proof proceeds in several steps. The main steps are similar to~\cite{borwein1997polynomials}, except that we adapt the polynomials constructed in~\cite{borwein1997polynomials}, which have coefficients in $\{0,1,-1\}$, to polynomials with coefficients in $\{-H,...,H\}$.
\paragraph{Step 1.}
We parametrize the degree $N$ by an integer $m \in \N$. 

By~\cite[Theorem 2.7]{BEK1999}, we can find a polynomial $F$ of degree at most $m-1$ with coefficients in $\{0,+1,-1\}$ and with multiplicity $L \geq 1$ at $1$, where $L \geq c_{1} \sqrt{m/\log(m)}$ for some $c_1>0.$ 
Besides, by~\cite[Theorem 2.4]{BEK1999} we know that we must have $L \leq c_{2} \sqrt{m}$.
Let $D:= \lfloor 3em/L \rfloor + 1$ and consider the following polynomial
\[G(x) \coloneqq H \cdot x^m \cdot (x^{mD}-1) F(x^D)\frac{x^D-1}{x-1} - F(x) \; . \]
We write $N$ for the degree of $G$, we note that $N \leq 2mD+m-1 < 3mD$ and that $N \geq mD = c m^{3/2}$ (for a constant $c >0$) so that $N \rightarrow + \infty$ for $m \rightarrow + \infty$ . We also note that by construction $G$ has coefficients in $\{-H,...,H\}$. We will show that $G$ has a root very close to $1$.
\begin{lemma}\label{lem:beta root G}
For $m$ large enough, $G$ has a real root $\alpha>1$ such that 
\[\frac{1}{2HmD^{L+2}} \leq \alpha - 1 \leq \frac{2}{HmD^{L+2}} \; . \]
\end{lemma}

Lemma~\ref{lem:beta root G} follows from the intermediate value theorem and the next key lemma.
\begin{lemma}\label{lem:key lemma}
Let $a:= F^{(L)}(1)/L!$ with $a\in \Z,|a|\geq 1$ by assumption. Then, we can write $G(x) = a \cdot (x-1)^L \cdot \phi(x)$, where $\phi(x) = -1 + H m D^{L+2} (x-1) + E(x)$ for some function $E(x)$ such that for $x_{t} = 1 + \frac{t}{HmD^{L+2}}$ we have
\begin{equation}\label{eq:phi change of sign}
\phi(x_{1/2}) = - \frac{1}{2} + E(x_{1/2}),\;  \phi(x_{2}) = 1 + E(x_{2}),\; |E(x_{1/2})| \leq 1/4,\text{ and }|E(x_{2})| \leq 1/4 \; .
\end{equation}
\end{lemma}
From the above lemma we get a sign change for $\phi$ and hence for $G$, very close to $1$ on the right, which is enough to conclude.
\begin{proof}[Proof of Lemma~\ref{lem:key lemma}.]
A Taylor expansion at $1$ gives
\[
G(x) = a \cdot (x-1)^L \cdot \left( -1 + H m D^{L+2} (x-1) + E(x) \right)
\]
with $\displaystyle
E(x) = - \frac{F^{(L+1)}(1)}{a (L+1)!}(x-1) + \sum_{r = 2}^\infty \frac{G^{(L+r)}(1)}{a (L+r)!}(x-1)^r$. 
For $\phi(x) = -1 + H m D^{L+2} (x-1) + E(x)$ we indeed get $G(x) = a \cdot (x-1)^{L} \cdot \phi(x)$, and 
 for $x_{t} = 1 + \frac{t}{HmD^{L+2}}$ we have
 \[\phi(x_{1/2}) = - 1 + \frac{1}{2} + E(x_{1/2}) = - \frac{1}{2} + E(x_{1/2})\;  \text{ and }\;  \phi(x_{2}) = - 1 + 2 + E(x_{2}) = 1 + E(x_{2})\; .
\]
There remains to show that we can control $|E(x_t)|$ for $t$ small enough:
\begin{itemize}
\item $F$ has degree at most $m-1$ and coefficients bounded by $1$ so 
\begin{equation}\label{eq:bound F}
F^{(r)}(1) \leq 1 \cdot m \cdot (m-1)^r \leq m^{r+1} \; .
\end{equation}
\item $G$ has degree $N$ and coefficients bounded by $H$ so $
G^{(r)}(1) \leq H \cdot (N+1)  \cdot (N)^r \leq 2 H N^{r+1}$.
\end{itemize}
We start by bounding the first term of $E(x)$ in absolute value. From $|a| \geq 1$ and \eqref{eq:bound F} we have
\begin{equation*}
|\frac{F^{(L+1)}(1)}{a (L+1)!}(x-1)| \leq \frac{m^{L+2}}{(L+1)!}|x-1| \; .
\end{equation*}
We now bound the second term of $E(x)$ in absolute value:
\begin{align*}
|\sum_{r = 2}^\infty \frac{G^{(L+r)}(1)}{a (L+r)!}(x-1)^r | & \leq \frac{2 H |x-1|^2 (3mD)^{L+3}}{(L+2)!} \sum_{r = 2}^\infty \frac{(3mD|x-1|)^{r-2}}{\frac{(L+r)!}{(L+2)!}} \\
& \leq \frac{2 H |x-1|^2 (3mD)^{L+3}}{(L+2)!} \sum_{r = 2}^\infty (3mD|x-1|)^{r-2} \; .
\end{align*}
Now for $|x-1| < \frac{1}{2 \cdot 3mD}$ we get that the sum in the right-hand side of the above equation is smaller than $2$, so we conclude that for $|x-1| < \frac{1}{2 \cdot 3mD}$ we have
\begin{equation}\label{eq:bound on E(x)}
|E(x)| \leq \frac{m^{L+2}}{(L+1)!}|x-1| + \frac{4 H  (3mD)^{L+3}}{(L+2)!} |x-1|^2 \; .
\end{equation}
Let $x_t = 1 + \frac{t}{HmD^{L+2}}$ for $t \in [1/2,2]$. We can check that $|x_t-1| < \frac{1}{6 mD}$, and now we bound $E(x_t)$:
\begin{itemize}
\item For the first term of $E(x_t)$ as in \eqref{eq:bound on E(x)} we have
\begin{align*}
\frac{m^{L+2}}{(L+1)!}|x_t-1| = \frac{m^{L+2}}{(L+1)!}\frac{t}{HmD^{L+2}}  = \frac{t}{H} \frac{m^{L+1}}{(L+1)!}\frac{1}{D^{L+2}} & = \frac{t}{ D H} \frac{1}{(L+1)!}\left(\frac{m}{D}\right)^{L+1} \\
& \leq \frac{t}{ D H} \frac{1}{3^{L+1}} \frac{1}{(L+1)!}\left(\frac{L}{e}\right)^{L+1}
\end{align*}
where the equalities are just reformulations and the inequality is from $D \geq 3em/L$ .
Therefore, for $L$ large enough (and therefore for $m$ large enough), from $k! \geq (k/e)^{k}$ for any $k \in \N$, we get $\displaystyle \frac{m^{L+2}}{(L+1)!}|x_t-1| \leq \frac{t}{ D H} \frac{1}{3^{L+1}}.$
\item For the second term of $E(x_t)$ as in \eqref{eq:bound on E(x)} we have (absorbing the factor $4$ in the constants later)
\begin{align*}
\frac{H  (3mD)^{L+3}}{(L+2)!} |x_t-1|^2 = \frac{H  (3mD)^{L+3}}{(L+2)!} \frac{t^2}{H^2 m^2 D^{2L+4}} & = \frac{t^2}{H} \frac{3^{L+3}}{(L+2)!} \frac{m^{L+1}}{D^{L+1}}
 \\
& \leq \frac{t^2}{H} \frac{9}{(L+2)!} \left( \frac{L}{e}\right)^{L+1}
\end{align*}
where the inequality follows from $m/D \leq L/3e$. Now from $k! \geq (k/e)^k$ we get that $(L+2)! \sim \sqrt{2\pi(L+2)} \left(\frac{L+2}{e}\right)^{L+2}$, hence we conclude that for $L$ large enough we must have, for some $C'>0$,
\[\frac{t^2}{H} \frac{9}{(L+2)!} \left( \frac{L}{e}\right)^{L+1} \leq \frac{t^2}{H}\frac{9}{L+2} \left( \frac{e}{L+2}\right)^{L+2} \left( \frac{L}{e} \right)^{L+1} \leq \frac{C'}{HL} t^2.\]
\end{itemize}
Overall, we have obtained that for $L$ large enough (and therefore $m$ large enough) we have that, for some absolute constants $C,C'>0$,
$\displaystyle |E(x_t)| \leq \frac{C}{HD} + \frac{C'}{HL} \leq \frac{C}{D} + \frac{C'}{L}$ for any $t \in [1/2,2]$ (using $H\geq 1$).
Therefore, for $L$ large enough, we can ensure that $|E(x_{1/2})|  \leq 1/4$ and $|E(x_{2})|  \leq 1/4$,
which combined with \eqref{eq:phi change of sign} concludes the proof of Lemma~\ref{lem:key lemma}. 
\end{proof}
We have shown the existence of a real root $\gamma$ of $G$, with $\gamma >1$ and $\gamma - 1 = O \left( \frac{1}{HmD^{L+2}} \right)$. 
\paragraph{Step 2.} We now compute an equivalent of $mD^{L+2}$. The next lemma is shown in \cite{borwein1997polynomials}, at the end of their proof of Theorem 5. We include it here for completeness.
\begin{lemma}\label{lem:small denominator}
Let $m \in \N$ and $(L_m)_{m \in \N}$, $(D_m)_{m \in \N}$, and $(N_m)_{m \in \N}$ be three sequences of integers such that there exists some $c_1,c_2>0$ satisfying
\[c_1 \sqrt{m/\log(m)} \leq L_m \leq c_2 \sqrt{m},\; D_m = \lfloor 3em/L_m\rfloor +1,\text{ and } N_m = 2mD_m+m-1\;.\]
Then, there exists a constant $c_3>0$ such that for $m$ large enough we have \[\frac{1}{mD_m^{L_m+2}} \leq \exp(- c_3 (N_m \log(N_m))^{1/3}) \; . \]
\end{lemma}
\begin{proof}
We omit the subscript $m$ in the proof for readability.
We have, for $m$ large enough,
\[\log(mD^{L+2}) \geq L \log(D) \geq L \log(m/L)  \geq L \log(\sqrt{m}/c_2) \geq c' L \log(m) 
\]
where the first inequality drops the $\log(m)$ term and the second uses the definition of $D$ and drops a term $L \log(3e)$. Now $c' L \log(m)   \geq c'' \sqrt{m/\log(m)}\log(m)\geq c'' \sqrt{m \log(m)}$ for some constant $c''>0$.
 This shows $\displaystyle \frac{1}{mD^{L+2}} \leq \exp\left( - c (m\log(m))^{1/2}\right).$
Now we note that
\[N < 3mD < 3m (6em/L) \leq \frac{18e m^2}{L} \leq c_1 m^{3/2} (\log(m))^{1/2}\]
where we use the definition of $D$ in the second strict inequality, and the lower bound on $L$ for the last inequality. This shows, for $m$ large enough, that $\displaystyle(N \log(N))^{1/3} \leq c_4 \sqrt{m \log(m)}$ which concludes the proof. 
\end{proof}
\paragraph{Last step.} 
The polynomial $G$ constructed above has a root $\gamma_{N,H}>1$ with $\gamma_{N,H}-1 <  \frac{1}{H} \exp \left( - c (N \log(N))^{1/3}\right)$ for some $c>0$. Consider $P_{N,H}(x) = x^{N} G(x^{-1})$. This polynomial has integer coefficients, bounded in absolute values by $H$. Additionally, $\alpha_{N,H} \coloneqq \gamma^{-1}_{N,H}$ is a real root of $P_{N,H}$ with $\alpha_{N,H} \in (0,1)$, and we conclude that $1-\alpha_{N,H} < \frac{1-\alpha_{N,H}}{\alpha_{N,H}} = \gamma_{N,H}-1 < \frac{1}{H} \exp \left( - c' (N \log(N))^{1/3}\right)$ for some constant $c'>0$.
 This concludes the proof. 
\section{Proof for \Cref{sec:deter d sens}.}\label{app:deter d sens}

We use the following theorem.
\begin{theorem}[\cite{BEK1999}]\label{th:bek99}
Let $p(x)=\sum_{k=0}^{n}c_k x^k$ be a nonzero polynomial such that $\max_{k \in \{0,...,n\}} |c_k| \leq 1$ and $c_0 \neq 0$. Then, the multiplicity of $1$ as a root of $p(x)$ cannot exceed $c \sqrt{n (1-\log(|c_0|))}$, where $c>0$ is a universal constant.
\end{theorem}
\begin{proof}[Proof of Lemma~\ref{lem:BEK99 at 0}]
We have $F(\rho) = \frac{1}{(1+\rho)^{j}}Q\left(\frac{1}{1+\rho}\right)$ with $Q(y) \coloneqq \sum_{k=j}^{N} a_k y^{k-j}$. This shows that the multiplicity of $F$ at $0$ is equal to the multiplicity of the polynomial $Q$ at $1$. Note that $Q$ has degree bounded by $N-j$ and $Q(0)=a_j \neq 0$. Consider $\tilde{Q}=Q/H$, which has the same multiplicity at $1$ as $Q$. We can apply \Cref{th:bek99} to bound the multiplicity of $\tilde{Q}$ at $1$ by $A \sqrt{(N-j)(1+\log(H/|Q(0)|))}$, for some absolute constant $A>0$. 
We always have $j \geq 0$ and since $F$ has integer coefficients we have $Q(0) \in \Z$. Since $Q(0) \neq 0$, we have $|Q(0)| \geq 1$. Hence, we conclude that the multiplicity of $F$ at $0$ is bounded by $A \sqrt{N(1+\log(H))}$. 
\end{proof}
\subsection{Proof of \Cref{prop:comp vector height and degree}.}
We recall that $ C^{\sigma,\sigma'}_{i}(\rho) = r^{\sigma(i)}_{i} + (v^{\sigma'}_\rho)_{\sigma(i)} - (1+\rho)(v^{\sigma'}_{\rho})_i$. To prove Proposition~\ref{prop:comp vector height and degree} we first analyze the properties of value functions in deterministic MDP instances. It is easier to analyze every term as a function of the discount factor $\alpha$ and then replace $\alpha$ by $(1+\rho)^{-1}$ at the end of our analysis.
\paragraph{Value functions in deterministic instances.}
In a deterministic MDP, every stationary deterministic policy $\sigma$ induces a weighted digraph where the vertices are the states of the MDP. An edge $(i,j)$ is present in the graph if $\sigma$ chooses in state $i$ an action that leads to a transition in state $j$ (recall that transitions are deterministic). The weight of a feasible $(i,j)$ is $r^{a}_{i} \in \Z$ (we use the abbreviation $r_{ij}$ in the rest of this section). The trajectory induced by $\sigma$ when starting from $i$ has the following ``lasso-shape" $``\pi\gamma"$, where $\pi$ is an elementary path (i.e., a sequence of pairwise distinct states), and $\gamma$ is an elementary circuit (i.e., a sequence of states where only the first and last are equal). More precisely, the last state of $\pi$ coincides with the first and last states of $\gamma$, and this state is the only one that appears more than once among the states of $\pi$ and of $\gamma$. Given a path $\tilde{\pi}=(i_0,\dots,i_{k})$, we set 
\[ 
 \langle r,\tilde{\pi} \rangle_{\alpha} \coloneqq  \alpha r_{i_0i_1} + \alpha^2 r_{i_1i_2} + \dots +  \alpha^{k} r_{i_{k-1}i_k} \; .
\] 
Using this notation we get that the value function satisfies 
\[v^{\sigma}_{i}(\alpha) = \langle r,\pi \rangle_{\alpha} + \alpha^p \frac{1}{1-\alpha^q} \langle r,\gamma\rangle_{\alpha} \; ,\]
where $p$ is the length $|\pi|$ of $\pi$ (i.e., the number of edges it contains) and $q$ is the length $|\gamma|$ of $\gamma$.  

The digraph induced by a stationary policy $\sigma$ (i.e., the digraph obtained by keeping only the edges that can be attained by this policy) can be decomposed as follows. This digraph comprises disjoint elementary circuits $\gamma_1, \dots,\gamma_r$. Moreover, for each $h\in\{1, \dots ,r\}´$, we denote by $U_h$ the set of nodes from which there is a directed path to circuit $\gamma_h$. The sets $U_1,\dots,U_r$ constitute a partition of the node set. The digraph may be visualized as a collection of ``river networks'', each $U_h$ being a river network going to the ``sea'' $\gamma_h$. We arrive at the next lemma. 
\begin{lemma}\label{Lemma_eq_value_path_circuit}
\begin{enumerate}
\item[(a)] For any $i\in U_h$, we have
\begin{equation}\label{eq_value_path_circuit}
v_i^{\sigma}(\alpha ) = \p<r,\pi_h> + \frac{\alpha^{p_h}\p<r,\tilde{\gamma}_h>}{1-\alpha^{q_h}} \; , \end{equation}
where $\pi_h$ is an elementary path contained in $U_h$, $\tilde{\gamma}_h$ is a cyclic conjugate of the circuit $\gamma_h$, $p_h\coloneqq  |\pi_h|$ and $q_h\coloneqq |\gamma_h|$. Moreover, $p_h\leq n_h-1$ and $q_h + p_h\leq n_h$, where $n_h\coloneqq |U_h|$.
\item[(b)] The function $\left(1-\alpha^{q_h} \right) v^{\sigma}_{i}(\alpha)$ is of the form $\sum_{k=0}^{n_{h}} a_k \alpha^k$ where $|a_k| \leq 2 W$, $a_{0}=0$, $|a_{1}|\leq W$, and $|a_{q_{h}}|\leq W$.
\end{enumerate}
\end{lemma}

\begin{proof}
\begin{enumerate}
\item[(a)] The equality~\eqref{eq_value_path_circuit} follows from the discussion above. Since the last state of $\pi_h$ coincides with the first and last states of $\tilde{\gamma}_h$, and this state is the only one that appears more than once in $\pi_h$ and $\tilde{\gamma}_h$, we have $p_h\leq n_h-1$ and $q_h + p_h\leq n_h$.
\item[(b)] From the first statement we get that $\displaystyle \left( 1-\alpha^{q_h} \right) v^{\sigma}_{i}(\alpha)= \left(1-\alpha^{q_h}\right)\langle r,\pi_h \rangle_{\alpha} + \alpha^{p_h} \langle r,\tilde{\gamma}_h \rangle_{\alpha}.$
If $\pi_h=(i_0,\dots,i_{p_h})$ and $\tilde{\gamma}_h=(j_0,\dots,j_{q_h})$, the right-hand side in the above equality is equal to 
\[ \sum_{k = 1}^{p_h} \alpha^{k} r_{i_{k-1}i_k} - \sum_{k = 1}^{p_h} \alpha^{q_h+k}r_{i_{k-1}i_k} + \sum_{k = 1}^{q_h} \alpha^{p_h + k}r_{j_{k-1}j_k} \; .
\]
Then, since the rewards are integers, $\left(1-\alpha^{q_h} \right) v^{\sigma}_{i}(\alpha)$ has integer coefficients as a polynomial in $\alpha$. Besides, this polynomial has the indicated (maximal possible) degree because $p_h+q_h \leq n_h$ by item (a). Finally, the claim on the absolute value of the coefficients follows from the fact that in the first term above the exponents of $\alpha$ range in $\{1,\ldots ,p_h \}$, while in the third term they range in $\{p_h+1,\ldots , p_h+q_h\}$. 
\end{enumerate}
\end{proof}
\paragraph{Analysis of the comparison vector.}
We can now reformulate the comparison vectors. We continue to work with $\alpha \in (0,1)$ instead of $\rho>0$ for now. For $\alpha = 1/(1+\rho)$, we have $\displaystyle C_{i}^{\sigma,\sigma'}(\rho)  = r^{\sigma(i)}_{i} + (v^{\sigma'}_{\rho})_{\sigma(i)} - (1+\rho)(v^{\sigma'}_{\rho})_i = r^{\sigma(i)}_{i} + v^{\sigma'}_{\sigma(i)}(\alpha) - \alpha^{-1}v^{\sigma'}_{i}(\alpha)
$.

Suppose that $i \in U_h$ and $\sigma(i) \in U_\ell$. Then, by \Cref{Lemma_eq_value_path_circuit}, we can assume that $(1-\alpha^{q_h})v^{\sigma'}_{i}(\alpha)$ is of the form $\sum_{k=0}^{n_{h}} a_k \alpha^k$ and $(1-\alpha^{q_\ell})v^{\sigma'}_{\sigma(i)}(\alpha)$ is of the form $\sum_{k=0}^{n_{\ell}} a'_k \alpha^k$. 
Let us define $\alpha \mapsto \tilde{F}(\alpha)$ and $\alpha \mapsto \tilde{G}(\alpha)$ as
\[\tilde{G}(\alpha) \coloneqq \prod_{q \in \{q_h,q_{\ell} \}} (1-\alpha^q) \; \text{ and } \; \tilde{F}(\alpha)  \coloneqq \tilde{G}(\alpha) \left( r^{\sigma(i)}_{i} + v^{\sigma'}_{\sigma(i)}(\alpha) - \alpha^{-1}v^{\sigma'}_{i}(\alpha)  \right)\; .
\]
We next analyze $\tilde{F}$ and $\tilde{G}$.
\begin{itemize}
\item {\em Analysis of $\tilde{G}$.} We consider two cases. If $q_h = q_{\ell}$ we have $\tilde{G}(\alpha) = 1-\alpha^{q_h}$ and $\tilde{G}$ has height $1$ and degree $q_h$, otherwise $\tilde{G}(\alpha) = 1 - \alpha^{q_h} - \alpha^{q_{\ell}} + \alpha^{q_h + q_{\ell}}$, which has height $1$ and degree $q_h + q_{\ell}$. For $\delta \coloneqq 1_{\{q_{h} \neq q_{\ell}\}}$ we get that $\deg(\tilde{G}) = q_{h} + \delta q_{\ell}$, from which we get $\deg(\tilde{G}) \leq n_{h} + \delta n_{\ell} \leq n$. 
\item {\em Analysis of $\tilde{F}$.}
The coefficient of $\alpha^k$ in $\tilde{F}(\alpha)$ is
\[b_k = r^{\sigma(i)}_{i}\left(1_{\{k=0\}} - 1_{\{k=q_h\}} - 1_{\{k=q_{\ell}\}} + 1_{\{k=q_{h} + q_{\ell}\}} \right) + a_{k} - a_{k-q_{h}} - a'_{k+1} + a'_{k+1-q_{\ell}} \; .\]
For $k \notin \{0,q_h,q_{\ell},q_{h} + q_{\ell}\}$, we have that $|b_k|$ is bounded by a sum of four terms, each bounded by $2W$ by \Cref{Lemma_eq_value_path_circuit}, and so $|b_k| \leq 8W$.
Assume that $q_{h} \neq q_{\ell}$. By \Cref{Lemma_eq_value_path_circuit} we have:
\begin{itemize}
    \item[-] The exponent of order $0$ has a coefficient of at most $2W$.
    \item[-] The exponent of order $q_{h}$ has coefficient at most $2W + 4W + W = 7W$.
    \item[-] The exponent of order $q_{\ell}$ has coefficient at most $3W + 4W + W = 8W$.
    \item[-] The exponent of order $q_{h}+q_{\ell}$ has coefficient at most $3W + 4W + W = 8W$.
\end{itemize}
The same analysis for the case $q_{h}=q_{\ell}$ gives $H(\tilde{F})\leq 4W$ in this case. So we always have $H(\tilde{F})\leq 8W$. For the degrees, we note that $\tilde{G}(\alpha) v^{\sigma'}_{\sigma(i)}(\alpha)$ has degree at most $n_{\ell} + \delta q_{h} \leq n$, while $\tilde{G}(\alpha) \alpha^{-1}v^{\sigma'}_{i}(\alpha)$ has degree at most $n_{h} -1 + \delta q_{\ell} \leq n$. Hence $\deg(\tilde{F}) \leq n$.
\end{itemize}
Now we define $F(\rho) \coloneqq  \tilde{F}((1+\rho)^{-1})$ and $G(\rho) \coloneqq  \tilde{G}((1+\rho)^{-1})$. The only remaining thing to prove is that the order of $G$ at $0$ is at most two. This follows from $1-\alpha^q = 1 - (1+\rho)^{-q} = q \rho + O(\rho^2)$ and therefore the order of $G$ is at most $|\{q_h,q_\ell\}| \leq 2$. 
This concludes the proof.

We show that the bound $8W$ is tight using the following deterministic MDP instance. There are seven states $\{s_1,...,s_7\}$. From state $s_1$ we can pick action $a_1$ or $a_2$. 
If we pick action $a_1$, we transition to $s_2$ and obtain a reward $-W$, then to $s_3$ with reward $+W$, then to $s_4$ with reward $+W$. From $s_4$ we transition deterministically to $s_3$ with reward $-W$.
If we pick action $a_2$ in $s_1$ we transition to state $s_5$ and obtain reward $+W$, then to state $s_6$ with reward $-W$, then to state $s_7$ with reward $-W$. The state $s_7$ is absorbing with reward $+W$. Take $\sigma'$ choosing $a_1$ in $s_1$, and $\sigma$ choosing $a_2$ in $s_1$.
In this instance we have $q_{s_5}=1,q_{s_1}=2$, $G(\alpha) = (1-\alpha^2)(1-\alpha)$, and $(1-\alpha^{q_{s_1}})v^{\sigma'}_{s_1}(\alpha) = W (-\alpha + \alpha^2 + 2 \alpha^3 - 2 \alpha^4)$ and $(1-\alpha^{q_{s_5}})v^{\sigma'}_{s_5}(\alpha) = W (-\alpha + 2 \alpha^3)$. Hence 
\[\tilde{F}(\alpha) =  W (2 - 4 \alpha - 2 \alpha^2 + 8 \alpha^3 - 2 \alpha^4 - 2 \alpha^5)\]
which shows that the bound $8W$ from Proposition~\ref{prop:comp vector height and degree} is tight.

\end{document}